\documentclass[final,3p,times]{elsarticle}

\usepackage{amssymb,amsmath,amsthm, amsfonts}
\usepackage{algorithmic,algorithm}
\usepackage{graphicx,graphics}
\usepackage{underlin}
\usepackage{color}
\usepackage{enumerate}
\usepackage{subfigure}
\usepackage{relsize}

\newtheorem{theorem}{Theorem}
\newtheorem{lemma}[theorem]{Lemma}
\newtheorem{prop}[theorem]{Proposition}
\newtheorem{defi}{Definition}
\newtheorem{coro}{Corollary}
\newtheorem{remark}{Remark}
\newtheorem{example}{Example}

\newcommand*{\Go}{G_{\overline\varphi}}
\newcommand*{\vp}{\varphi}
\newcommand*{\vpo}{\overline\varphi}
\newcommand*{\mc}{\mathcal}
\newcommand{\eqcl}{\mathrel{\mathsmaller{\mathsmaller{^{\boldsymbol{\sqsubseteq}}}}}}

\renewcommand{\mod}{\ensuremath{\operatorname{mod}}}

\DeclareMathOperator*{\BOX}{\Box}

\journal{Preprint}

\begin{document}

\sloppy

\begin{frontmatter}

\title{The Relaxed Square Property}

\author[SB]{Marc Hellmuth} 
\ead{mhellmuth@bioinf.uni-sb.de}
\author[LJU,LEI]{Tilen Marc}
\ead{marct15@gmail.com}
\author[LEI]{Lydia Ostermeier\footnote{corresponding author}}
\ead{glydia@bioinf.uni-leipzig.de}
\author[LEI,MPI,IZI,TBI,SFI]{Peter F.\ Stadler}
\ead{studla@bioinf.uni-leipzig.de}

\address[SB]{Center for Bioinformatics, 
	      Saarland University, Building E 2.1, 
	      D-66041 Saarbr{\"u}cken, Germany
             }
\address[LJU]{Faculty of Mathematics and Physics, 
	     University of Ljubljana, SI-1000 Ljubljana, Slovenia
             }
\address[LEI]{Bioinformatics Group, Department of Computer Science,
              and Interdisciplinary Center for Bioinformatics,
              University of Leipzig,
              H{\"a}rtelstra{\ss}e 16-18, D-04107 Leipzig, Germany
             }
\address[MPI]{Max Planck Institute for Mathematics in the Sciences,
              Inselstra{\ss}e 22, D-04103 Leipzig, Germany
             }
\address[IZI]{RNomics Group, Fraunhofer Institut f{\"u}r Zelltherapie 
              und Immunologie -- IZI Perlickstra{\ss}e 1, 
              D-04103 Leipzig, Germany
             }
\address[TBI]{Department of Theoretical Chemistry
              University of Vienna,
              W{\"a}hringerstra{\ss}e 17, A-1090 Wien, Austria}
\address[SFI]{Santa Fe Institute, 1399 Hyde Park Rd., Santa Fe, NM 87501,
              USA}

\begin{abstract}
  Graph products are characterized by the existence of non-trivial
  equivalence relations on the edge set of a graph that satisfy a so-called
  square property. We investigate here a generalization, termed
  \emph{RSP-relations}. The class of graphs with non-trivial RSP-relations
  in particular includes graph bundles. Furthermore, RSP-relations are
  intimately related with covering graph constructions.  For $K_{2,3}$-free
  graphs finest RSP-relations can be computed in polynomial-time.  In
  general, however, they are not unique and their number may even grow
  exponentially. They behave well for graph products, however, in
  sense that a finest RSP-relations can be obtained easily from finest
  RSP-relations on the prime factors.
\end{abstract}


\end{frontmatter}


\section{Introduction}

Modern proofs of prime factor decomposition (PFD) theorems for the
Cartesian graph product rely on characterizations of the product relation
$\sigma$ on the edge set of the given graph \cite{Imrich:94}. The key
property of $\sigma$ is that connected components of the subgraphs induced
by the classes of $\sigma$ are precisely the layers, i.e., $(e,f)\in
\sigma$ if and only if the edges $e$ and $f$ belong to copies of the same
(Cartesian) prime factor \cite{Sabidussi:60,Hammack:11a}. Classical results
in the theory of graph products establish that $\sigma$ can be derived from
other, easily computable, relations on the edge set: \[ \sigma =
\mathfrak{C}(\delta) = (\theta \cup\tau)^*, \] where $\mathfrak{C}(\delta)$
denotes the convex closure of the so-called $\delta$-relation and $(\theta
\cup\tau)^*$ is the transitive closure of two different relations known as
the Djokovi\'{c}-Winkler relation $\theta$ and relation $\tau$
\cite{Imrich:94,Hammack:11a}.

Of particular interest for us is the relation $\delta$. An equivalence
relation $R$ is said to have the \emph{square property} if (i) any pair of
adjacent edges which belong to distinct equivalence classes span a unique
chordless square and (ii) the opposite edges of any chordless square belong
to the same equivalence class. The importance of $\delta$ stems from the
fact that it is the unique, finest relation on $E(G)$ with the square
property.

An equivalence relation has the \emph{unique square property} if any two
adjacent edges $e$ and $f$ from distinct equivalence classes span a unique
chordless square with opposite edges in the same equivalence class.  The
slight modification, in fact a mild generalization, of the relation
$\delta$ turned out to play a fundamental role for the characterization of
graph bundles \cite{Zmazek:02b} and forms the basis of efficient algorithms
to recognize Cartesian graph bundles \cite{Imrich:97,Zmazek:02a}. Graph
bundles \cite{Pisanski:83}, the combinatorial analog of the topological
notion of a fiber bundle \cite{Husemoller:93}, are a common generalization
of both Cartesian products \cite{Hammack:11a} and covering graphs
\cite{Abello:91}.

The key distinction of the unique square property is that, in contrast to
the square property, opposite edges do not have to be in the same
equivalence class for all chordless squares. Any such relation that is in
addition weakly 2-convex yields the structural properties of a graph bundle
\cite{Zmazek:02b}. Moreover, every Cartesian graph bundle over a
triangle-free simple base can be characterized by the relation $\delta^*$,
which satisfies the unique square property \cite{Imrich:97}. In a recent
attempt to better understand the structure of equivalence relations on the
edge set of a graph $G$ that satisfy the unique square property, we
uncovered a surprising connection to equitable partitions on the vertex set
of $G$ \cite{HOS14:EquiParty} and a Cartesian factorization of certain
quotient graphs that was previously observed in the context of quantum
walks on graphs \cite{Bachman:12}. It was shown that for any equivalence
class $\varphi$ of a relation $R$ with unique square property the connected
components of the graph $G_{\overline\varphi}=(V(G),E(G)\setminus\varphi)$
form a natural \emph{equitable} partition
$\mathcal{P}_{\overline\varphi}^R$ of the vertex set of $G$. Moreover, the
so-called common refinement $\mathcal{P}^R$ of this partitions
$\mathcal{P}^R_{\vpo}$ yields again an equitable partition of $V(G)$ and
the quotient $G/\mc P^R$ has then a product representation as
$G/\mc{P}^R\cong\Box_{\vp \sqsubseteq R}
G_{\varphi}/\mc{P}^R_{\overline{\varphi}}$.

In \cite{OstermeierL:14}, it was shown that a further relaxation of the
unique square property to the relaxed square property still retains the
product decomposition of these quotient graphs.  The connected components
of $G_\varphi=(V(G),\varphi)$ have a natural interpretation as fibers,
while the graph $G_{\vpo}/\mc P^R_\varphi$ can be seen as base graph. Such
a decomposition is a graph bundle if and only if edges in $G$ linking
distinct connected components of $G_\varphi$ induce an isomorphism between
them. Thus, graphs with this type of relations on the edge set, which we
call \emph{RSP-relations} for short, are a natural generalization of graph
bundles.

In this contribution we will examine RSP-relations more systematically.
First we show that, as in the case of the unique square property, there is
no uniquely determined finest RSP-relation for given graphs in
general. Even more, the number of such finest relations on a graph can grow
exponentially. However, we will see that the finest RSP-relations $R$ are
``bounded'' by relations $\delta_0, \delta_1$ and $\tau$ so that $(\tau
\cup \delta_1)^* \subseteq R\subseteq \delta_0^*$. We explain how (finest)
RSP-relations can be determined in certain graph products, given the
RSP-relations in the factors. The main difficulty in determining finest
RSP-relations derive from $K_{2,3}$ as induced subgraphs. We provide a
polynomial-time algorithm for $K_{2,3}$-free graphs and give a recipe how
finest RSP-relations can be constructed in complete and complete bipartite
graphs. Finally, we examine the close connection of covering graphs and
RSP-relations.

\section{Preliminaries}
\label{sect:prelim}

\paragraph{Notation}

In the following we consider finite, connected, undirected, simple graphs
unless stated otherwise. A graph $G$ has vertex set $V=V(G)$ and edge set
$E=E(G)$. A graph $H$ is a subgraph of $G$, $H\subseteq G$, if
$V(H)\subseteq V(G)$ and $E(H)\subseteq E(G)$. A subgraph $H$ is an \emph{induced
  subgraph} of $G$ if $x,y\in V(H)$ and $[x,y]\in E(G)$ implies $[x,y]\in
E(H)$. $H$ is called \emph{spanning subgraph} if $V(H)=V(G)$.  If none of
the subgraphs $H$ of $G$ is isomorphic to a graph $K$, we say that $G$ is
\emph{K-free}. A subgraph $H=(\{a,b,c,d\},\{[a,b], [b,c], [c,d], [a,d]\})$
is called \emph{square}, will often be denoted by $a-b-c-d$ and we say that
$[a,b]$ and $[c,d]$, resp., $[b,c]$ and $[a,d]$ are \emph{opposite}
edges. The \emph{complete} graph on $n$ vertices is denoted by $K_n$ and
the \emph{complete bipartite} graph on $n+m$ vertices by $K_{m,n}$.

We will consider equivalence relations $R$ on $E$ and denote equivalence
classes of $R$ by Greek letters, $\varphi\subseteq E$. We will furthermore
write $\varphi \eqcl R$ to indicate that $\varphi$ is an equivalence class
of $R$. The complement $\vpo$ of an $R$-class $\vp$ is defined as
$\overline{\varphi} := E\setminus \varphi$. For an equivalence class
$\varphi\eqcl R$, an edge $e$ is called $\vp$-edge if $e\in \vp$. The
subgraph $G_{\varphi}$ has vertex set $V(G)$ and edge set $\varphi$. The
connected components of $G_{\varphi}$ containing vertex $x\in V(G)$ are
called \emph{$\varphi$-layer through $x$}, denoted by $G_{\varphi}^x$.
Analogously, the subgraphs $G_{\vpo}$ and $G_{\vpo}^x$ are defined.  Two
$\varphi$-layer $G_\varphi^x, G_\varphi^y$ are said to be \emph{adjacent},
if there exists an edge $[x',y']\in\overline\varphi$ with $x'\in
V(G_\varphi^x)$ and $y'\in V(G_\varphi^y)$.

An equivalence relation $Q$ is \emph{finer} than a relation $R$ while the
relation $R$ is \emph{coarser} than $Q$ if $(e,f)\in Q$ implies $(e,f)\in
R$, i.e, $Q\subseteq R$. In other words, for each class $\vartheta$ of $R$
there is a collection $\{ \chi | \chi\subseteq \vartheta\}$ of $Q$-classes,
whose union equals $\vartheta$. Equivalently, for all $\varphi\eqcl Q$ and
$\psi\eqcl R$ we have either $\varphi\subseteq \psi$ or
$\varphi\cap\psi=\emptyset$. If $R$ is not an equivalence relation, then we
will denote with $R^*$ the finest equivalence relation that contains $R$.
Moreover, an equivalence relation $R$ is \emph{non-trivial} if it has
at least two equivalence classes. 

For a given partition $\mc{P}=\{V_1,\dots,V_l\}$ of $V(G)$ of a graph $G$,
the \emph{quotient graph} $G/\mc{P}$ has as its vertex set $\mc P$ and
there is an edge $[A,B]$ for $A,B\in\mc{P}$ if and only if there are
vertices $a\in A$ and $b\in B$ such that $[a,b]\in E(G)$.
A partition $\mathcal{P}$
of the vertex set $V(G)$ of a graph $G$ is \emph{equitable} if, for all
(not necessarily distinct) classes $A,B\in\mathcal{P}$, every vertex $x\in
A$ has the same number
$ m_{AB} :=  |N_G(x) \cap B| $
of neighbors in $B$.

\paragraph{Graph Cover and Homomorphisms}
A homomorphism $f: G\rightarrow H$ between two graphs $G$ and $H$ is called
\emph{locally surjective} if $f(N_G(u))=N_H(f(u))$ for all vertices $u\in
V(G)$, i.e., if $f_{|N_G(u)}:N_G(u)\rightarrow N_H(f(u))$ is a
surjection. We use here the obvious notation $N_G(v)$ for the open
neighborhood of $v$ in the graph $G$.  Analogously, $f$ is called
\emph{locally bijective} if for all vertices $u\in V(G)$ it holds that
$f(N_G(u))=N_H(f(u))$ and $|f(N_G(u))|=|N_H(f(u))|$, i.e.,
$f_{|N_G(u)}:N_G(u)\rightarrow N_H(f(u))$ is a bijection.  Notice, a
locally surjective homomorphism $f: G\rightarrow H$ is already globally
surjective if $H$ is connected.  If there exists a locally surjective
homomorphism $f: G\rightarrow H$, we call $G$ a \emph{quasi-cover} of
$H$. Locally surjective homomorphisms are also known as role colorings
\cite{EB:91}.  A locally bijective homomorphism is called a \emph{covering
  map}.  $G$ is a \emph{(graph) cover} or \emph{covering graph} of $H$ if
there exists a covering map from $G$ to $H$, in which case we say that $G$
\emph{covers} $H$.  $|V(H)|$ is then a multiple of $|V(G)|$, i.e., $|V(H)|=
k |V(G)|$. $H$ is referred to as \emph{$k$-fold cover} of $G$. Moreover,
every covering map $f:H\rightarrow G$ satisfies $|f^{-1}(u)|=k$ for all
$u\in V(G)$ \cite{Fiala:08}.  For more detailed information about locally
constrained homomorphisms and graph cover we refer to
\cite{FPT:08,Fiala:08}.

\paragraph{Graph Products} 
There are three associative and commutative standard graph products, the
\emph{Cartesian product} $G\BOX H$, the \emph{strong product} $G\boxtimes
H$, and the \emph{direct product} $G\times H$, see \cite{Hammack:11a}.

All products have as vertex set the Cartesian set product $V(G)\times
V(H)$. Two vertices $(g_1,h_1)$, $(g_2,h_2)$ are adjacent in $G\boxtimes H$
if $(i)$ $[g_1,g_2]\in E(G)$ and $h_1=h_2$, or $(ii)$ $[h_1,h_2]\in E(G_2)$
and $g_1 = g_2$, or $(iii)$ $[g_1,g_2]\in E(G)$ and $[h_1,h_2]\in E(G_2)$.
Two vertices $(g_1,h_1)$, $(g_2,h_2)$ are adjacent in $G\Box H$ if they
satisfy only $(i)$ and $(ii)$, while these two vertices are adjacent in
$G\times H$ if they satisfy only $(iii)$. 

Every finite connected graph $G$ has a decomposition $G=\BOX_{i=1}^n G_i$,
resp., $G=\boxtimes_{i=1}^n G_i$ into prime factors that is unique up to
isomorphism and the order of the factors \cite{Sabidussi:60}. For the
direct product an analogous result holds for non-bipartite connected
graphs.

The mapping $p_i: V(\Box_{i=1}^n G_i)\rightarrow V(G_i )$ defined by
$p_i(v) = v_i$ for $v = (v_1,v_2,\ldots,v_n)$ is called \emph{projection}
on the $i$-th factor of $G$. By $p_i(W) = \{p_i(w)\mid w\in W\}$ the set of
projections of vertices contained in $W\subseteq V(G)$ is denoted. An
equivalence relation $R$ on the edge set $E(G)$ of a Cartesian product
$G=\Box_{i=1}^n G_i$ of (not necessarily prime) graphs $G_i$ is a
\emph{product relation} if $(e,f) \in {R}$ if and only if there
exists a $j\in \{1,\ldots,n\}$ such that $|p_j(e)|=|p_j(f)|=2$. The
$G_i$-layer $G_i^w$ of $G$ is then the induced subgraph with vertex set
$V(G_i^w)=\{v\in V(G)\mid p_j(v)=w_j, \text{ for all } j\neq i\}$. It is
isomorphic to $G_i$.

Given two graphs $G$ and $H$, a map $p:G\to H$ is called a \emph{graph map}
if $p$ maps adjacent vertices of $G$ to adjacent or identical vertices in
$B$ and edges of $G$ to edges or vertices of $B$.  A graph $G$ is a
\emph{(Cartesian) graph bundle} if there are two graphs $F$, the fiber, and
$B$ the base graph, and a graph map $p:G\to B$ such that: For each vertex
$v\in V(B)$, $p^{-1}(v)\cong F$ and for each edge $e\in E(B)$ we have
$p^{-1}(e)\cong K_2\square F$.

\section{RSP-Relations: Definition and Basic Properties}

As mentioned in the introduction, relations that have the \emph{square
  property} play a fundamental role for the $\Box$-PFD of graphs. In
particular, the relation $\delta$ is the unique, finest relation on $E(G)$
with the square property. For such relations two incident edges of
different classes span exactly one chordless square and this square has
opposite edges in the same equivalence classes. A mild generalization of
the latter kind of relations are relations that have the \emph{unique
  square property}. Here two incident edges $e$ and $f$ of different
classes might span more than one square, however, there must be exactly one
chordless square spanned by $e$ and $f$ with opposite edges in the same
equivalence classes. As it turned out, a further generalization of such
relations plays an important role for the characterization of certain
properties of hypergraphs \cite{OstermeierL:14}.  Here, we examine 
this generalization in realm of undirected graph in a systematic manner.

\begin{defi}
  Let $R$ be an equivalence relation on the edge set $E(G)$ of a connected
  graph $G$. We say $R$ has the \emph{relaxed square property} if any two
  adjacent edges $e,f$ of $G$ that belong to distinct equivalence classes
  of $R$ span a square with opposite edges in the same equivalence class of
  $R$.
\end{defi}
An equivalence relation $R$ on $E(G)$ with the relaxed square property will
be called an \emph{RSP-relation} for short. In contrast to the more
familiar (unique) square property,we do not require there that squares
spanned by incident edges that belong to different equivalence classes are
unique or chordless.

The following basic result was shown in \cite{OstermeierL:14} for hypergraphs and
equivalence relations with the ``grid property'', of which graphs and
RSP-relations are a special case.

\begin{lemma}[\cite{OstermeierL:14}]
  \label{lem:classes}
  Let $R$ be an RSP-relation on $E$ of a connected graph $G=(V,E)$. Then
  each vertex of $G$ is incident to at least one edge of each $R$-class and
  thus, the number of $R$-classes is bounded by the minimum degree of $G$.
  Moreover, if $S$ is a coarser equivalence relation, $R \subseteq S$, then
  $S$ is also an RSP-relation.
\end{lemma}

\begin{figure}[tbp]
 \centering
  \subfigure[][]{
    \label{fig:Labelname1}
     \includegraphics[bb= 236 653 359 741,
     width=0.2\textwidth]{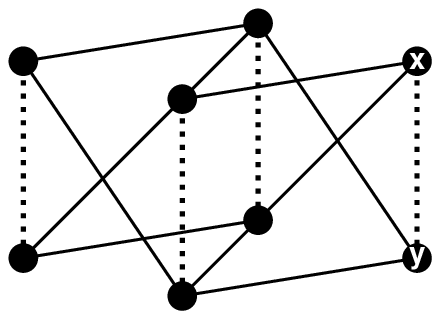} 
     \hspace{50pt}
     \includegraphics[bb= 236 653 359 741, width=0.2\textwidth]{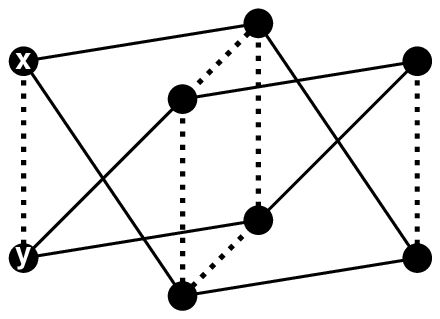}
  } 
  \begin{center}
  \subfigure[][]{
    \label{fig:2}
     \includegraphics[bb= 123 653 472 741,width=0.6\textwidth]{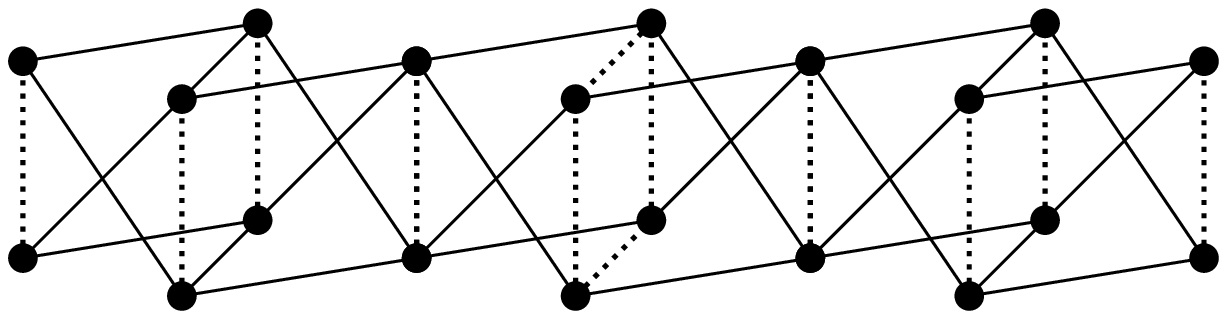}
  }
  \end{center}
  \caption{In Fig.\ \subref{fig:Labelname1} two isomorphic graphs with two
    non-equivalent finest RSP-relations are shown. Each RSP-relation has
    two equivalence classes, highlighted by dashed and solid edges.  By
    stepwisely identifying the vertices marked with $x$ and $y$, resp., one
    obtains a chain of graphs $G$, see Fig. \subref{fig:2}. For each
    subgraph that is a copy of the graph above, a finest RSP-relation can
    be determined independently of the remaining parts of the graph $G$.
    Hence, with an increasing number of vertices of such chains $G$ the
    number of finest RSP-relations is growing exponentially.  }
  \label{fig:expo}
\end{figure}

\begin{figure}[t]
  \begin{center}
    \includegraphics[bb=172 274 481 370,scale=0.75]{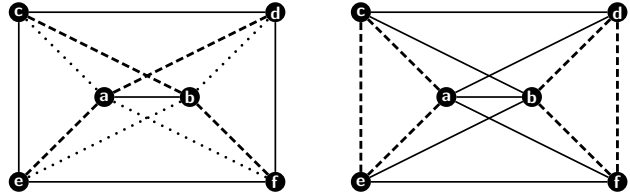}
\end{center}  
\caption{The two panels show two distinct finest RSP-relations $R$ and $S$
  on a graph with different number of equivalence classes, see Example
  \ref{exmp:diffNrEqcl}.}
\label{fig:NonuniqueRel2}
\end{figure}

For later reference we record the following technical result:
\begin{lemma} \label{lem:delete_classes} Let $R$ be an RSP-relation on 
  the edge set $E$
  of a connected graph $G=(V,E)$ and $\varphi$ be an equivalence class of
  $R$. Moreover, let $S$ be the equivalence relation on the edge set
  $E\setminus\varphi$ of the spanning subgraph $G'=(V,E\setminus\varphi)$
  of $G$ that retains all equivalence classes $\psi\neq\varphi$ of
  $R$. Then $S$ is an RSP-relation.
\label{lem:subgraphRSP}
\end{lemma}
\begin{proof}
  Let $e,f$ be adjacent edges in $E(G')$ such that $(e,f)\notin S$, say
  $e\in\psi, f\in\psi'$, $\varphi\neq\psi,\psi'\eqcl S\subseteq R$.  By
  construction, $e,f\in E(G)$ and $(e,f)\notin R$.  Thus, there exists a
  square with edges $e,f,e',f'$ such that $e,e'$ and $f,f'$ are opposite
  edges and $e'\in\psi$ as well as $f'\in\psi'$.  Hence, $e',f'\in E(G')$
  and thus the assertion follows.
\end{proof}

The RSP-relation $S$ on the spanning subgraph, as defined in Lemma
\ref{lem:subgraphRSP}, need not to be a finest RSP-relation, although $R$
might be a finest one. Consider the right graph in Figure
\ref{fig:NonuniqueRel2}. If $S$ consists only of the class $\vpo$ that is
highlighted by the drawn-through edges, then the spanning subgraph
$H=(V(G),E(G)\setminus\varphi)$ is the Cartesian graph product of a path on
three vertices and an edge. The finest RSP-relation on $E(H)$ is thus the
product relation $\sigma$ w.r.t.\ the unique $\Box$-PFD of $H$ with two
equivalence classes.

As the examples in Figures \ref{fig:expo}, \ref{fig:NonuniqueRel2} and
\ref{fig:NonuniqueRel} show, there is no unique finest RSP-relation for a
given graph $G$ and finest RSP-relations need not to have the same number
of equivalence classes. Even more, the number of such finest relations on a
graph can grow exponentially as the example in Figure~\ref{fig:expo} shows.

\begin{example}
  \label{exmp:diffNrEqcl}
  There are graphs $G=(V,E)$ with distinct finest RSP-relations that even
  have a different number of equivalence classes. Consider the graph in
  Figure \ref{fig:NonuniqueRel2}. We leave it to the reader to verify that
  the relations, whose equivalence classes are indicated by different line
  styles, indeed satisfy the relaxed square property. The RSP-relation on
  the left graph has three and on the right graph two equivalence
  classes. It remains to show, that both RSP-relations are finest ones.
  \par\noindent
  {\textbf{Left Graph:}} For all equivalence classes there is a vertex
  that is incident to exactly one edge of each class. Lemma
  \ref{lem:classes} implies that $R$ is finest RSP-relation.
  \par\noindent
  {\textbf{Right Graph:}} The equivalence class indicated by the
  dashed edges cannot be subdivided further since this would lead to
  vertices that are not met each of the two or more subclasses, thus
  contradicting Lemma~\ref{lem:classes}. The equivalence class depicted by
  drawn-through edges is isomorphic to a Cartesian product $P_3\Box K_2$.  Using
  Lemma~\ref{lem:delete_classes}, the only possible split would be the
  Product relation on this subgraph, i.e., with classes
  $\psi_1=\{[a,b],[c,d],[e,f]\}$ and $\psi_2=\{[a,d],[a,f],[b,c],[b,e]\}$.
  But then there is no square with opposite edges in the same equivalence
  classes spanned by the edges $[b,c]$ and $[c,e]$, again a contradiction.
\end{example}

We next discuss the relationship of (finest) RSP-relations with relations
of the edge set that play a role in the theory of product graphs and graph
bundles.
\begin{defi}[\cite{Feder92}]
Two edges $e = \{x,z \}$ and $f = \{z,y \}$ are in the \emph{relation
$\tau$}, $e\tau f$ if $z$ is the unique common neighbor of $x$ and $y$. 
\end{defi}
In other words, two edges are in relation $\tau$ if they are adjacent and
there is no square containing both of them. Obviously, $\tau$ is
symmetric. Its reflexive and transitive closure, i.e. the smallest
equivalence relation containing $\tau$, will be denoted by $\tau^*$. By
definition, $\tau^*\subseteq R$ for any RSP-relation $R$.

\begin{defi}
Two edges $e,f\in E(G)$ are in the \emph{relation $\delta_0$}, $e\delta_0 f$, 
if one of the following conditions is satisfied:
\begin{itemize}
\item[(i)]   $e$ and $f$ are opposite edges of a square.
\item[(ii)]  $e$ and $f$ are adjacent and there is no square 
             containing $e$ and $f$, i.e. $(e,f)\in\tau$. 
\item[(iii)] $e=f$.
\end{itemize}   
\end{defi}
The relation $\delta_0$ is reflexive and symmetric. Its transitive closure,
denoted with $\delta_0^*$, is therefore an equivalence relation. 

\begin{prop} \label{prop:finest_rel} 
  Let $G$ be a connected $K_{2,3}$-free graph and $R$ an equivalence
  relation on $E(G)$. Then $R$ has the relaxed square property if and only
  if $\delta_0\subseteq R$.
\end{prop}
\begin{proof}
  It is easy to see, that $\delta_0^*$ has the relaxed square property and
  moreover, that any equivalence relation containing $\delta_0$ has the
  relaxed square property.
  
  Let $R$ be an RSP-relation on the edge set of a connected $K_{2,3}$-free
  graph $G$. Notice, if $G$ contains no $K_{2,3}$ than any pair of adjacent
  edges of $G$ span at most one square. Let $e,f$ be two edges in $G$ such
  that $(e,f)\in\delta_0$. We have to show that this implies $(e,f)\in
  R$. If $e=f$, then $(e,f)\in R$ is trivially fulfilled since $R$ is an
  equivalence relation. If $e$ and $f$ are not adjacent, they have to be
  opposite edges of a square. Let $g$ be an edge of this square, that is
  adjacent to both edges $e$ and $f$. If $e$ and $g$ are not in relation
  $R$, by the relaxed square property, they span some square with opposite
  edges in the same equivalence class.  Since $G$ contains no $K_{2,3}$,
  this square is unique, thus $(e,f)\in R$. Assume now, $(e,g)\in R$. If
  $e$ and $f$ are not in the same equivalence class of $R$, we can conclude
  that also $f$ and $g$ are in distinct equivalence classes, since $R$ is
  an equivalence relation.  Thus, by the relaxed square property, $f$ and
  $g$ span a square with opposite edges in the same equivalence class and
  as $G$ is $K_{2,3}$-free, this square has to be unique, which implies
  $(e,f)\in R$, a contradiction. Now let $e$ and $f$ be two adjacent edges
  and suppose for contraposition $(e,f)\notin R$. Hence, $e$ and $f$ have
  to span a square. Thus, condition (ii) in the definition of $\delta_0$ is
  not satisfied, hence, $(e,f)\notin \delta_0$. In summary, we can conclude
  $\delta_0\subseteq R$.
\end{proof}

Proposition~\ref{prop:finest_rel} implies that there is a uniquely
determined finest RSP-relation, namely the relation $\delta_0^*$ if $G$ is
$K_{2,3}$-free.  However, if $G$ is not $K_{2,3}$-free, there is no
uniquely determined finest RSP-relation, see Fig.\ \ref{fig:expo},
\ref{fig:NonuniqueRel2} and \ref{fig:NonuniqueRel}. Moreover, the quotient
graphs that are induced by these relations (see
\cite{HOS14:EquiParty,OstermeierL:14}) need not to be isomorphic.

\begin{figure}[t]
\begin{center}
  \includegraphics[bb= 59 255 502 512, 
   scale=0.75]{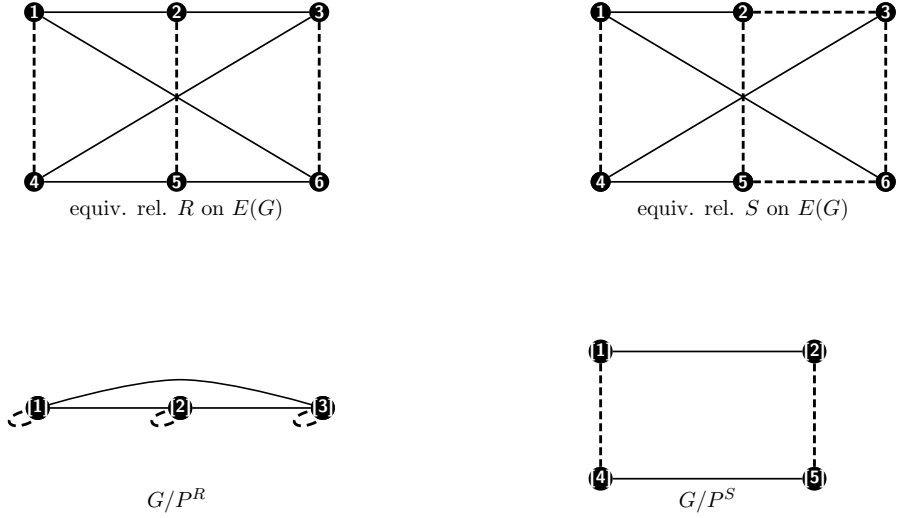}
\end{center}  
\caption{Two distinct RSP-relations $R$ and $S$ on the edge set of the same
  graph $G$ and the quotient graphs induced by these relations
  (below). Their coarsest common refinement, i.e., the coarsest equivalence
  relation $T$ with $T \subseteq R$ and $T\subseteq S$ does not have the
  relaxed square property. Moreover, the quotient graphs induced by these
  relations need not to be isomorphic.}
\label{fig:NonuniqueRel}
\end{figure}

By construction, $\delta_0$ places all edges of a $K_{2,3}$-subgraph in the
same equivalence class. In many graphs this leads to an RSP-relation which
is not finest. On the other hand, the opposite edges of a square that is
not contained in a $K_{2,3}$ must always be in the same equivalence
class. This motivates us to introduce the following
\begin{defi}
  Two edges $e,f\in E(G)$ are in the \emph{relation $\delta_1$}, $e\delta_1
  f$, if one of the following conditions is satisfied:
\begin{itemize}
\item[(i)] $e$ and $f$ are opposite edges of a square that is not contained
  in any $K_{2,3}$ subgraph of $G$.
\item[(ii)] $e=f$.
\end{itemize}   
\end{defi}

If $G$ is $K_{2,3}$-free then it is easy to verify that $\delta_0 = (\tau
\cup \delta_1)$. Proposition \ref{prop:finest_rel} implies that $(\tau \cup
\delta_1)^*$ is contained in any RSP-relation and therefore, that it is a
uniquely determined finest RSP-relation on $K_{2,3}$-free graphs. We can
summarize this discussion of the properties of finest RSP-relations as
follows:
\begin{theorem}
  Let $G$ be an arbitrary graph and $R$ be a finest RSP-relation on
  $E(G)$. Then it holds that:
  \[(\tau \cup \delta_1)^* \subseteq R\subseteq \delta_0^*.\]
  Moreover, if $G$ is $K_{2,3}$-free, then $(\tau \cup \delta_1)^* = R =
  \delta_0^*$.
  \label{thm:tauRdelta}
\end{theorem}

  Theorem~\ref{thm:tauRdelta} suggests that $K_{2,3}$-subgraphs are to
  blame for complications in understanding RSP-relations. It will therefore
  be useful to consider a subclass of RSP-relations that are
  ``well-behaved'' on $K_{2,3}$-subgraphs. They will turn out to play a
  crucial role to establish the connection of RSP-relations,
  (quasi-)covers, and equitable partitions. We fix the notation for
  $K_{2,3}$ so that $\{x,y\},\{a,b,c\}$ is the canonical partition of of
  the vertex set. We say that graph $K_{2,3}$ has a \emph{forbidden
    coloring} if the edges $[a,x]$, $[x,c]$, and $[y,b]$ are in one
  equivalence class $\varphi$ and the other edges are in the union
  $\vpo$ of the classes different from $\vp$.

\begin{defi}
  An RSP-relation is \emph{well-behaved} (on $G$) if $G$ does not contain a
  subgraph isomorphic to a $K_{2,3}$ with a forbidden coloring.
\end{defi}

For a graph $G$ and an RSP-relation $R$ consisting of only two
equivalence classes we can strengthen this definition.  It is easy to
verify that in this case the two statements are equivalent:
\begin{itemize}
\item[(i)] $R$ is well-behaved
\item[(ii)] for each pair of incident edges $[a,b]$, $[a,c]$ which are not
  in relation $R$ there exists a unique (not necessarily chordless) square
  $a-b-d-c$ with opposite edges the same classes, i.e., $([a,b], [c,d]),
  ([a,c], [b,d])\in R$.
\end{itemize}
  In the general case $(i)$ implies $(ii)$. To see this, note that if
  there are incident edges that span more than one square, say
  $\textrm{SQ}_1$ and $\textrm{SQ}_2$, with opposite edges in the same
  classes, then there is a $K_{2,3}$ with forbidden coloring that consists
  of the squares $\textrm{SQ}_1$ and $\textrm{SQ}_2$. Hence, $R$ cannot be
  well-behaved. The converse is not true in general, as shown in
  Fig.~\ref{fig:CounterUnion}.  by the non-well-behaved RSP-relation $R'$
  that nevertheless has property $(ii)$.

To obtain well-behaved RSP-relations $R$ on $G$ one can simply use
$\delta_0$ and coarsenings of it. That is, \emph{any} equivalence relation
$R$ with $\delta_0 \subseteq R$ is well-behaved.  In this case, all edges
of any $K_{2,3}$-subgraph are in the same equivalence class.  However,
coarsenings of arbitrary well-behaved RSP-relation $R$ need not be
well-behaved, see Fig.~\ref{fig:CounterUnion}.

\begin{figure}[tbp]
  \centering
  \includegraphics[bb= 250 590 400 730, scale=0.75]{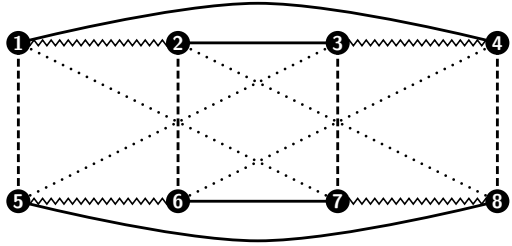}
  \caption{The well-behaved RSP-relation $R$ on the edge set $E(G)$ of the
    ``diagonalized cube'' $G$ has the four equivalence classes
    $\vp_1,\vp_2,\vp_3$ and $\vp_4$ depicted by solid, zigzag, dotted and
    dashed edges, respectively. In addition, $R$ satisfies the unique
    square property. The relation $R'$ with classes $\vp_3,\vp_4$ and
    $\psi_1=\vp_1\cup\vp_2$, however, is not well-behaved, because the
    $K_{2,3}$-subgraph with partition $\{1,6\}$ and $\{2,4,5\}$ has a
    forbidden coloring. Note, $R'$ has the unique square property.}
  \label{fig:CounterUnion}
\end{figure}

Furthermore, if $R$ is not well-behaved, this is
equivalent to the existence of squares with two adjacent edges in same
class $\varphi\eqcl R$ and others in class(es) different from $\varphi$,
see Figure~\ref{fig:forbiddensubs} and the next explanations. It is easy to
verify that any $K_{2,3}$(-subgraph) with a forbidden coloring contains
such a square. By way of example, consider the square $a-x-c-y$ in
Figure~\ref{fig:forbiddensubs}.  Conversely, let $R$ be an RSP-relation on
$E(G)$ and suppose that $G$ contains a square $a-x-c-y$ with
$([a,x],[c,x])\in R$ and $([a,x],[c,y]),([a,y],[c,x])\notin R$.  By the
relaxed square property, $[a,x]$ and $[c,y]$ span a square, say $a-x-b-y$
with opposite edges in the same equivalence class. Hence, there is a
complete bipartite graph $K_{2,3}$ with partition $\{x,y\} $ and
$\{a,b,c\}$ of $V(K_{2,3})$ and forbidden coloring.

\begin{figure}[htbp]
 \centering
 \includegraphics[bb= 214 401 311 498,
 width=0.2\textwidth]{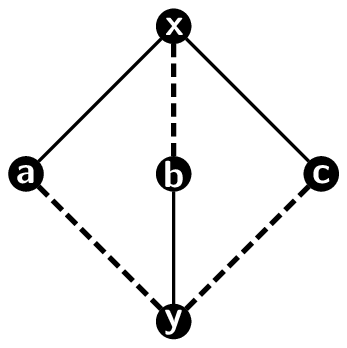}
 \hspace{55pt}
 \includegraphics[bb= 214 401 311
 498,width=0.2\textwidth]{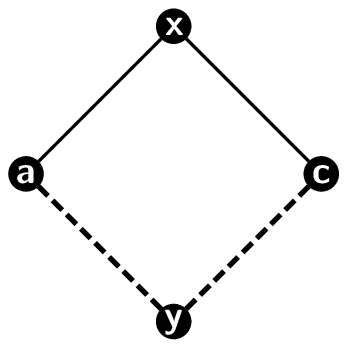}
 \caption{Forbidden coloring of a (sub)graph isomorphic to $K_{2,3}$ based
   on the classes $\vp$ and $\vpo$ of a (non-well-behaved)
   RSP-relation. The class $\vpo$ might consist of more than one
   equivalence class.  The existence of a forbidden coloring is equivalent
   to the existence of squares spanned by edges in same equivalence class
   with opposite edges in different equivalence classes. Such a square
   contained in the (sub)graph $K_{2,3}$ is shown on the right.}
  \label{fig:forbiddensubs}
\end{figure}

Let us now turn to the computational aspects of RSP-relations. It is an easy
task to determine finest relations that have the square property in
polynomial time, see \cite{HIK-13,HIK-14}.  In contrast, it seem to be hard
in general to determine one or all finest RSP-relations. We conjecture that
the corresponding decision problem is NP- or GI-hard \cite{GJ79,KST} for
general graphs.

\begin{algorithm}[tbp]
  \caption{\texttt{Compute RSP-Relation }}
  \label{alg:rexSP}
  \begin{algorithmic}[1]
     \STATE \textbf{INPUT:} A connected graph $G=(V,E)$
     \STATE Compute $R_0=(\delta_1\cup\tau)^*$; \label{alg:tau}
	  \STATE $Q \gets \{(e,f) \mid e,f\in E, e\cap f\neq \emptyset \} 
                  \setminus R_0$;
     \STATE $j\gets 0$;
     \STATE \COMMENT{Note, edges $e$ and $f$ with $(e,f)\in Q$ are adjacent, 
                     span a square and are necessarily distinct}
     \WHILE{$Q\neq \emptyset$}
        \STATE Take an arbitrary pair $(e,f)\in Q$ with 
           $e\cap f\neq \emptyset$; 
        \STATE Let $sq_1, \dots, sq_k$ be the squares spanned 
           by $e$ and $f$;
        \STATE Find the opposite edges $e_i$ of $e$ and $f_i$ of $f$ in $sq_i$;
        \IF{there is a square $sq_i$ with $(e,e_i)\in R_j^*$ and 
                   $(f,f_i)\in R_j^*$ } \label{alg:check1}
           \STATE $Q\gets Q\setminus \{(e,f), (f,e)\}$; \label{alg:rem1}
        \ELSE
           \STATE take an arbitrary square, say $sq_1$ 
                   \COMMENT{with edge set $E_0=(e,f,e_1,f_1)$};
           \STATE $R_{j+1}\gets R_j^*\cup 
                  \{(e,e_1),(e_1,e),(f_1,f),(f,f_1)\}$;
           \STATE compute $R_{j+1}^*$; 
           \STATE $Q\gets Q\setminus R_{j+1}^*$;
           \STATE $j\gets j+1$;
        \ENDIF
     \ENDWHILE			
     \STATE $R\gets R_j^*$
     \STATE \textbf{OUTPUT:} An RSP-relation $R$ on $E$;
   \end{algorithmic}
\end{algorithm}

On the other hand, an efficient polynomial-time solution exists for
$K_{2,3}$-free graphs since $\delta_0$ can be constructed efficiently,
e.g., by listing all squares \cite{CN85}.  Algorithm \ref{alg:rexSP} serves
as a heuristic to find a finest RSP-relation for general graphs. The basic
idea is to start from the lower bound $R=(\delta_1\cup \tau)^*$ and to
unite equivalence classes of $R$ stepwisely until an RSP-relation is
obtained.

\begin{prop}
  Let $G=(V,E)$ be a given graph with maximum degree $\Delta$.  Algorithm
  \ref{alg:rexSP} computes an RSP-relation $R$ on $E$ in $O(|V||E|^2\Delta^4)$
  time.  If $G$ is $K_{2,3}$-free, then Algorithm \ref{alg:rexSP} computes
  a finest RSP-relation on $E$.
\end{prop}
\begin{proof}
  Clearly, $(\delta_1\cup \tau)^*$ must be contained in every RSP-relation
  $R$.  The set $Q$ contains all adjacent candidate edges $(e,f)$, where we
  have to ensure that they span a square with opposite edges in the same
  equivalence class.  Since we computed already $\tau$, we can conclude
  that if $e$ and $f$ are contained in $Q$, then they span some square.
  Thus, we check in Line \ref{alg:check1} whether there are already
  opposite edges $e'$ of $e$ and $f'$ of $f$ in one of those squares
  spanned by $e$ end $f$ with $(e,e'),(f,f') \in R_j^*$, i.e., $e'$ and
  $e$, resp., $f'$ and $f$ are in the same equivalence class.  If so, we
  can safely remove $(e,f)$ from $Q$.  If not, 
  we will construct a square  spanned by $e$ and $f$  with opposite edges in
  the same class and the 
  pair $(e,f)$ will be removed from $Q$ in the next run of the while-loop
  (Line \ref{alg:rem1}). 
  To be more precise, we take one of those
  squares spanned by $e$ and $f$ 
  and add $(e,e')$ and $(f,f')$ to $R_j$ resulting in
  $R_{j+1}$. Hence, $e$ and $f$ span now a square with opposite edges in
  the same class.  We then compute the transitive closure $R_{j+1}^*$. This
  might result in new pairs $(a,b)\in R_{j+1}^*$ of adjacent edges, which
  can safely be removed from $Q$ since they are in the same equivalence
  class, and thus do not need to span a square with opposite edges in the
  same class. Hence, we compute $Q\gets Q\setminus R_{j+1}^*$.  
  When $Q$ is empty all adjacent pairs (which span
  at least one square) are added in a way that at least one square has
  opposite edges in the same equivalence class. Thus, $R$ satisfies the
  relaxed square property.  Note, if $G$ is $K_{2,3}$-free, then all pairs
  $(e,f)$ of adjacent edges $e$ and $f$ already span a square with opposite
  edges in the same class, due to $\delta_1$.  Hence, all such pairs
  $(e,f)$ will be removed from $Q$, without adding any new pair to $R_0^*$.
  In this case we obtain $R=(\delta_1\cup \tau)^*$.

  In order to determine the time complexity we first consider the relation
  $\delta_1$. Note that there are at most $O(|E|\Delta^2)$ squares in a
  graph, that can be listed efficiently in $O(|E|\Delta)$ time, see Chiba
  and Nishizeki \cite{CN85}. For the computation of $\delta_1$, we have to
  check for each square $a-b-c-d$ whether it is contained in a $K_{2,3}$
  subgraph or not. Thus, we need to verify whether $a$ and $c$ have a
  common neighbor $x\not\in \{b,d\}$, and, if $b$ and $d$ have a common
  neighbor $x\not\in \{a,c\}$, respectively. If none of the cases occur,
  i.e., the square is not part of a $K_{2,3}$ subgraph, then we put the
  pairs $([a,b], [c,d])$ and $([a,d], [b,c])$ to $\delta_1$. This task can
  be done in $O(\Delta^2)$ time for each square, resulting in an overall
  time complexity of $O(|E|\Delta^4)$. The relation $\tau$ can be computed
  in $O(|V||E|)$ time \cite[Prop. 23.5]{Hammack:11a} and the transitive
  closure $(\delta_1\cup \tau)^*$ in $O(|E|^2)$ time, \cite[Prop.
  18.2]{Hammack:11a}. Thus, we end in time complexity
  $O(|E|^2\Delta^4)$ for the computation of $(\delta_1\cup \tau)^*$.
  Finally, we have to check for the at most $|V|\Delta^2$ pairs of adjacent
  edges whether they already span a square with opposite edges in the same
  class or not and compute the transitive closure $R_{j+1}^*$ if necessary.  
  Since there are at most $|E|\Delta^2$ squares, $|E|\leq
  |V|\Delta$, and the transitive closure can be computed in 
  $O(|E|^2)$ time, the latter task can be done in $O(|V||E|^2\Delta^3)$ time. 
\end{proof}

As the following example 
shows, the order in which the squares
are examined does matter in the general case, hence Alg.\ \ref{alg:rexSP}
does not produce a finest RSP-relation in general.  

\begin{example}
Consider the complete graph $K_5=(V,E)$ with vertex set 
$V=\mathbb Z_5$ and natural edge set.
After the init step we
have $R_0=\{(e,e)\mid e\in E\}$ and hence, $Q$ contains all pairs of
adjacent edges.
To obtain a finest RSP-relation, we could start with the pair
$([0,1][1,4])\in Q$ that span the square $0-1-4-3$ get as classes
$\vp_1 = \{[0,1],[3,4]\}$ and $\vp_2 = \{[1,4],[0,3]\}$ of $R_1^*$.
Continuing with $([0,1][1,2])\in Q$ and the square $0-1-2-3$, we obtain the
classes $\vp_1 \cup \{[2,3]\}$ and $\vp_2 \cup \{[1,2]\}$ of $R_2^*$.
Next, take $([0,1][0,4])\in Q$ and the square $0-1-2-4$, followed by the
pair $([0,1][0,2])\in Q$ and the square $0-1-4-2$, resulting in the classes
$\vp_1=\{[0,1][2,3],[3,4],[2,4]\}$ and
$\vp_2=\{[0,2],[0,3],[0,4],[1,2],[1,4]\}$ for $R_4^*$. Finally, take
$([0,1][1,3])\in Q$ and the square $0-1-3-4$ to obtain the classes $\vp_1$
and $\vp_2\cup\{[1,3]\}$ for a valid finest RSP-relation, see Example
\ref{exmpl:anotherKM} for further details. Note, the computed RSP-relation
is not well-behaved. 

However, if we start with the pair 
$([0,1][0,4])\in Q$ and square $0-1-3-4$,
followed by   
$([1,2][1,3])\in Q$ and $1-2-4-3$, then 
$([1,4][3,4])\in Q$ and $1-2-3-4$, next 
$([0,1][0,3])\in Q$ and $0-1-2-3$ and finally 
$([0,2][2,3])\in Q$ and $0-2-3-4$, the resulting
RSP-relation has only one equivalence class. 
\end{example}

\section{RSP-Relations and Graph Products}

Graph products are intimately related with the square property. It seem
natural, therefore to ask whether finest RSP-relations can be found more
easily in products. We use the symbol $\circledast$ for one of the three 
graph products defined in Section~\ref{sect:prelim}.

\begin{defi} \label{def:rel_prod}
  For $\circledast\in\{\Box,\boxtimes,\times\}$ let $G=\circledast_{i\in
  I}G_i$. For each $i\in I$ let $R_i$ be an equivalence relation on
  $E(G_i)$. Furthermore, define for $e\in E(G)$ the set $I_e:=\{i\in I\mid
  p_i(e)\in E(G_i)\}$. We define an equivalence relation $\circledast_{i\in
  I}R_i$ on $E(G)$ as follows: $(e,f)\in \circledast_{i\in I}R_i$ if and
  only if $I_e=I_f$ and $(p_i(e),p_i(f))\in R_i$, for all $i\in I_e$.
\end{defi}
If $\circledast=\Box$ then $|I_e|=1$ for all $e\in E(G)$, and if
$\circledast=\times$ then $I_e=I$ for all $e\in E(G)$.

\begin{lemma}
  \label{lem:rel_prod} For $\circledast\in\{\Box,\boxtimes,\times\}$ let
 $G=\circledast_{i\in I}G_i$. For each $i\in I$ let $R_i$ be an equivalence
 relation on $E(G_i)$. Then $R:=\circledast_{i\in I}R_i$ is an RSP-relation
 if and only if $R_i$ is an RSP-relation for all $i\in I$.
\end{lemma}

\begin{proof}
  First suppose $R_i$ has the relaxed square property for all $i\in I$.
  We have to show that $R$ has the relaxed square property.
  Therefore, let $e=[x,y],f=[x,z]\in E(G)$ such that $(e,f)\notin R$.
  We need to show that there exists a vertex $w\in V(G)$ such that 
  $e'=[w,z]\in E(G)$, $f'=[w,y]\in E(G)$ and $(e,e')\in R$ as well as
  $(f,f')\in R$.
  
  Let $I_0:=\{i\in I\mid (p_i(e),p_i(f))\in R_i\}$. Notice, that
  $I_0\subseteq I_e\cap I_f$. Moreover, we have $(p_j(e),p_j(f))\notin
  R_j$ for all $j\in (I_e\cap I_f)\setminus I_0=:I^*$. Since $R_i$ has the
  relaxed square property for all $i\in I$, for all $j\in I^*$ there exists
  a vertex $w_j\in V(G_j)$ such that $(p_j(e),[p_j(z),w_j])\in R_j$ as well
  as $(p_j(f),[p_j(y),w_j])\in R_j$.
  
  Let $w\in V(G)$ such that
  \begin{align*}
    p_i(w)&=p_i(x)\quad \text{for all }i\in I_0\\
    p_i(w)&=w_i\quad \text{\ \ \ \ for all }i\in I^*\\
    p_i(w)&=p_i(z)\quad \text{for all }i\in I\setminus I_e\\
    p_i(w)&=p_i(y)\quad \text{for all }i\in I\setminus I_f.
  \end{align*}
  Since $I=I_0\dot{\cup}I^*\dot{\cup}(I\setminus(I_e\cap I_f))$,
  $I\setminus (I_e\cap I_f)=I\setminus I_e \cup I\setminus I_f$ and
  $p_i(z)=p_i(x)=p_i(y)$ for all $i\in I\setminus I_e \cap I\setminus I_f$,
  this vertex exists in $V(G)$ and is well defined.
  
  We now have to verify that $w$ has the desired properties. More
  precisely, we have to verify the following statements:
  \begin{itemize}
  \item[(i)] $p_i(w)=p_i(z)$ for all $i\in I\setminus I_e$,
  \item[(ii)] $p_i(w)=p_i(y)$ for all $i\in I\setminus I_f$,
  \item[(iii)] $e_i':=[p_i(z),p_i(w)]\in E(G_i)$ such that
    $(p_i(e),e_i')\in R_i$ for all $i\in I_e$,
  \item[(iv)] $f_i':=[p_i(y),p_i(w)]\in E(G_i)$ such that $(p_i(f),f_i')\in
    R_i$ for all $i\in I_f$.
  \end{itemize}
  Assertions $(i)$ and $(ii)$ are trivially fulfilled by construction.  To
  prove assertion $(iii)$, note it holds that
  $I_e=I_0\dot{\cup}I^*\dot{\cup} I_e\setminus I_f$. From $p_i(w)=p_i(x)$
  for all $i\in I_0$, we conclude $e_i'=[p_i(z),p_i(x)]=p_i(f)\in E(G_i)$,
  and moreover, by construction of $I_0$ and since $R_i$ is an equivalence
  relation, we have $(p_i(e),e_i')\in R_i$ for all $i\in I_0$.  By the
  choice of $w$, it holds that $e_i'\in E(G_i)$ and $(p_i(e),e'_i)\in R_i$
  for all $i\in I^*$ . Finally, we have
  $e_i'=[p_i(z),p_i(y)]=[p_i(x),p_i(y)]=p_i(e)\in E(G_i)$ for all $i\in
  I_e\setminus I_f$ and since $R_i$ is an equivalence relation,
  $(e'_i,p_i(e))\in R_i$. Thus, $e'=[w,z]\in E(G)$ and $(e,e')\in R$.
 
  Assertion $(iv)$, which implies $f'=[w,y]\in E(G)$ and $(f,f')\in R$, can
  be shown by analogously.
  
  Now suppose $R$ is an RSP-relation. We have to show that for all $i\in
  I$, $R_i$ has the relaxed square property. Therefore, let $i\in I$ and
  $e_i=[x_i,y_i], f_i=[x_i,z_i]$ be two adjacent edges in $G_i$ such that
  $(e_i,f_i)\notin R_i$. We need to show, that there exists some vertex
  $w_i\in V(G_i)$ such that $e'_i:=[w_i,z_i], f'_i:=[w_i,y_i]$ are edges in
  $G_i$ with $(e_i,e_i')\in R_i$ and $(f_i,f_i')\in R_i$. By definition of
  $\circledast$, there exists edges $e=[x,y],f=[x,z]\in E(G),\;
  p_i(x)=x_i,p_i(y)=y_i,p_i(z)=z_i$, with $p_i(e)=e_i$ and $p_i(f)=f_i$,
  that are adjacent. It holds that $i\in I_e\cap I_f$ and by definition of $R$,
  $(e,f)\notin R$. Since $R$ has the relaxed square property, there exists
  some vertex $w\in V(G)$ such that $e':=[w,z], f':=[w,y]$ are edges in $G$
  with $(e,e')\in R$ and $(f,f')\in R$. That is, by definition of $R$,
  $I_e=I_{e'}$ and $(p_j(e),p_j(e'))\in R_j$ for all $j\in I_e$ as well
  as $I_f=I_{f'}$ and $(p_j(f),p_j(f'))\in R_j$ for all $j\in
  I_f$. Thus, we have in particular $(e_i,p_i(e')),(f_i,p_i(f'))\in R_i$
  and $z_i\neq p_i(w)\neq y_i$.  Moreover, $p_i(w)\neq x_i$, since
  otherwise $p_i(e')=[p_i(w),p_i(z)]=[x_i,z_i]=f_i$ and therefore
  $(f_i,e_i)=(p_i(e'),p_i(e))\in R_i$ must hold, a contradiction. Hence,
  with $w_i:=p_i(w)$ the assertion follows.
\end{proof}

\begin{figure}[t]
  \begin{center}
    \includegraphics[bb=149 403 340 597,scale=.8]{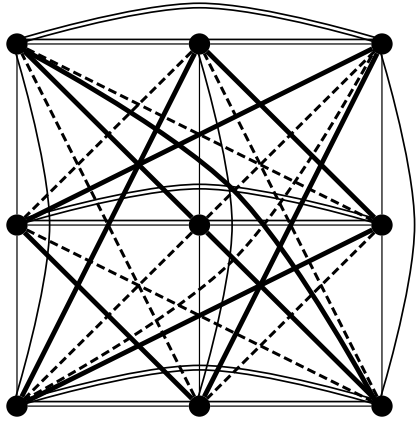}
  \end{center}  
  \caption{Refinement of product of relations of $K_9$ w.r.t.\ $K_9\cong
    K_3\boxtimes K_3$}
  \label{fig:K9}
\end{figure}

For $\circledast\in\{\times,\boxtimes\}$, the relation $R=\circledast_{i\in I}R_i$ need
not to be the finest RSP-relation on $E(G)=E(\circledast_{i\in I}G_i)$ although
$R_i$ is a finest RSP-relation on $E(G_i)$ for all $i\in I$. See
Fig.~\ref{fig:K9} for an example: Shown is the complete graph
$K_9$ with a finest RSP-relation consisting of four equivalence
classes depicted by drawn-through, double, dashed and thick lines. Joining the
two classes with dashed and thick edges to one class, one gets a coarser
relation $R_1\boxtimes R_2$, w.r.t. $K_9\cong K_3\boxtimes K_3$ where $R_i$
denotes the trivial relation on $E(K_3)$. This implies together with
Lemma~\ref{lem:delete_classes} that also $R_1\times R_2$ is not a finest
RSP-relation on $E(K_3\times K_3)$.

However, this does not hold for the Cartesian product $\Box$. Moreover, we
have:
\begin{lemma}
  \label{lem:finestCartRel}
  Let $G=\Box_{i\in I} G_i$ be a connected and simple graph. Then $R$ is a finest
  RSP-relation on $E(G)$ if and only if $R=\Box_{i\in I} R_i$ where $R_i$ is a
  finest RSP-relation on $E(G_i)$.
\end{lemma}
\begin{proof}
  First, observe the following: Let $R'$ be an arbitrary RSP-relation on
  $G$ and $[x,y], [y,z] \in E(G)$ incident edges that lie in the same layer
  of $G$, i.e. $p_j([x,y])\in E(G_j)$ and $p_j([y,z])\in E(G_j)$ for some
  $j\in I$. Moreover, let $[x,y]$ and $[y,z]$ be in different equivalence
  classes of $R'$. Since $R'$ is an RSP-relation, they lie on a four cycle
  $x-y-z-w$ with opposite edges in the same equivalence class. By the
  definition of the Cartesian product $w$ is also in the same layer as
  $x,y,z$, that is $w \in V(G_j^y)$. This shows that $R'$ limited to
  subgraph $G_j^x$ is also an RSP-relation.

  Let now $[x,y], [w,z] \in E(G)$ be such edges that lie on a four cycle
  $x-y-z-w$ with $j\in I$ such that $p_j([x,y])=p_j([w,z]) \in E(G_j)$. Assume
  that $[x,y]$ and $[w,z]$ do not lie in the same equivalence class of
  $R'$.  Then at least one of the pairs $[x,y], [x,w]$ or $[w,z], [x,w]$ do
  not lie in the same equivalence class of $R'$. Without loss of generality
  let $[x,y]$ and $[x,w]$ lie in different equivalence classes of $R'$. By
  the definition of the Cartesian product, $x-y-z-w$ is the only four cycle
  that contains $[x,y]$ and $[x,w]$. Since $R'$ is an RSP-relation, $[x,y]$
  and $[w,z]$ lie in the same equivalence class. By connectedness of $G$,
  all layers are connected. Therefore, all edges $\{[a,b]\in E(G): \,
  p_j([a,b])=p_j([x,y])\}$ are in the same equivalence class.

  Assume now that $R$ is a finest RSP-relation on $G$. We define relation
  $R_j$ on $E(G_j)$ for every $j\in I$ by $(e,f)\in R_j$ for $e,f\in E(G_j)$ if
  $p_j(e')=e$, $p_j(f')=f$ for some $e',f'\in E(G)$ and $(e',f')\in R$. By above
  arguments, this is an RSP-relation on $G_j$. Notice that $R$ corresponds
  to $\Box_{i\in I} R_i$ with possibly some joint equivalence classes, that
  emerge from different layers of $\Box_{i\in I} G_i$. Since $R$ is a
  finest RSP-relation, $R=\Box_{i\in I} R_i$. If $R_j$ is not a finest
  RSP-relation on $G_j$ for some $j\in I$, then the product of a finer
  relation on $G_j$ with $\Box_{i\in I\backslash \{ j \} } R_i$ is a finer
  relation as $R$, a contradiction.

  To see the converse, let $R=\Box_{i\in I} R_i$, where $R_i$ is a finest
  RSP-relation on $G_i$. If $Q$ is a finest relation on $G$, that is finer
  than $R$, by above arguments, $Q=\Box_{i\in I} Q_i$, where $Q_i$ is 
  finer or equal than $R_i$ for every $i\in I$. Thus $Q=R$.
\end{proof}

Lemma~\ref{lem:finestCartRel} implies not only that $R=\Box_{i\in I}R_i$ is
finest RSP-relation on $E(G)=E(\Box_{i\in I}G_i)$ if $R_i$ is finest
RSP-relation on $E(G_i)$, but also that any (finest) RSP-relation on a
Cartesian product graph must reflect the layer w.r.t. its (prime)
factorization.  However, this is not true for $\circledast=\boxtimes$, as
an example take $K_6\cong K_3\boxtimes K_2$ with the relation defined in
Example~\ref{exmpl:anotherKM}.

Following \cite{HOS14:EquiParty}, we introduce vertex partitions associated
with an equivalence relation $R$ on $E(G)$. In particular, we define for an
equivalence class $\vp\eqcl R$ the partitions 
\begin{equation*}
\mc P^R_{\varphi}:=\left\{V(G^x_{\varphi})\mid x\in V(G)\right\} 
\textrm{\ and\ }
\mc P^R_{\vpo}:=\left\{V(\Go^x)\mid x\in V(G)\right\}. 
\end{equation*} 

Graham and Winkler showed in \cite{GrahamWinkler:84} that the
Djokovi\'{c}-Winkler relation, or more precisely, the equivalence relation
$R=\theta^*$ on $E(G)$ induces a canonical isometric embedding of a graph
$G$ into a Cartesian product $\Box_{\varphi \eqcl R} G_{\varphi} /\mc
P_{\vpo}^{R}$.  Moreover, Feder \cite{Feder92} showed that if we choose
$R=(\theta \cup \tau)^*$ then $G \cong \Box_{\varphi \eqcl R} G_{\varphi}
/\mc P_{\vpo}^{R}$ and thus, $R$ coincides with the product relation
$\sigma$.

In \cite{OstermeierL:14}, we demonstrated that if $R$ is an RSP-relation
then
\begin{equation}
  G/\mc{P}^R\cong\Box_{\vp \sqsubseteq R} 
  G_{\varphi}/\mc{P}^R_{\overline{\varphi}},
  \label{eq:quotprod}
\end{equation}
where $\mc{P}^R$ denotes the common refinement of the partitions
$\mathcal{P}^R_{\vpo}$, $\varphi\eqcl R$, i.e.,
\begin{equation*}\label{eq:mcPR}
\mc{P}^R :=\left\{\bigcap_{\vp \eqcl R } 
  V( G_{\vpo}^x)\mid x\in V(G)\right\}\,,
\end{equation*}
which is again a partition of $V(G)$.
\begin{lemma}
  For $i\in I$ let $G_i$ be connected graphs and let $R_i$ be an
  RSP-relation on the edge set $E(G_i)$. Moreover, let
  $R:=\circledast_{i\in I} R_i$. It holds that:
  \begin{itemize}
  \item[($\Box$)] If  $G=\Box_{i\in I}G_i$ then
    $G/\mc P^R=\Box_{i\in I}G_i/\mc P^{R_i}$.
  \item[($\boxtimes$)]  If $G=\boxtimes_{i\in I}G_i$ then
    $G/\mc P^R=\mathcal L K_1$.
  \end{itemize}
  \label{lem:prod-quot}
\end{lemma}
  
\begin{proof}
  $(\Box)$ 
  By construction, $\psi$ is an equivalence class of $R$ if and only if
  there exists an $i\in I$ such that $p_i(e)\in E(G_i)$ and there exists
  $\varphi\in R_i$ with $p_i(e)\in\varphi$ for all $e\in
  \psi$. Hence, there exists a bijection $R=\circledast_{i\in I}
  R_i\rightarrow\dot{\bigcup}_{i\in I} R_i$. For $i\in I$ let
  $\varphi^i_1,\ldots,\varphi^i_{n_i}$ be the equivalence classes of $R_i$.
  Moreover, for $i\in I$ and $1\leq j\leq n_i$ let $\psi^i_j$ be the
  equivalence class of $R$ such that $I_e=\{i\}$ and $p_i(e)\in\varphi^i_j$
   for all $e\in\psi^i_j$. Thus, with Equation~\eqref{eq:quotprod}, we
  obtain $G/\mc P^R=\Box_{\psi\eqcl R}G_\psi/\mc{P}^R_{\overline\psi}=
  \Box_{i\in I}(\Box_{j=1}^{n_i}G_{\psi^i_j}/\mc{P}^R_{\overline{\psi^i_j}})$.
  Furthermore, due to Equation~\eqref{eq:quotprod}, we have $\Box_{i\in I}
  G_i/\mc P^{R_i}=\Box_{i\in I}(\Box_{j=1}^{n_i}
  {G_i}_{\varphi^i_j}/\mc P^{R_i}_{\overline{\varphi^i_j}})$.
      
  Hence, we need to show $G_{\psi^i_j}/\mc P^R_{\overline{\psi^i_j}}\cong
  {G_i}_{\varphi^i_j}/\mc P^{R_i}_{\overline{\varphi^i_j}}$ for all
  $i\in I$ and $1\leq j\leq n_i$, to prove the assertion. Therefore, we
  show that $G_{\overline{\psi_j^i}}(x)\mapsto
  {G_i}_{\overline{\varphi_j^i}}(p_i(x))$ for all $x\in V(G)$ defines an
  isomorphism $G_{\psi^i_j}/\mc P^R_{\overline{\psi^i_j}}\cong
  {G_i}_{\varphi^i_j}/\mc P^{R_i}_{\overline{\varphi^i_j}}$. If
  $G_{\overline{\psi_j^i}}(x)=G_{\overline{\psi_j^i}}(y)$, there exists a
  path $P_{x,y}:=(e_1,\ldots,e_k)$ from $x$ to $y$ in $G$, such that
  $e_l\notin\psi^i_j$ for $1\leq l\leq k$. Then
  $p_i(P_x,y)=(p_1(e_1),\ldots,p_i(e_k))$ is a walk from $p_i(x)$ to
  $p_i(y)$ in $G_i$ and by construction, it holds that
  $p_i(e_l)\notin\varphi^i_j$ for $1\leq l\leq k$, i.e.,
  ${G_i}_{\overline{\varphi_j^i}}(p_i(x))=
  {G_i}_{\overline{\varphi_j^i}}(p_i(y))$.  Thus, this mapping is well
  defined. Moreover, by the projection properties of a Cartesian product
  into its factors, this mapping is surjective. Now, suppose
  ${G_i}_{\overline{\varphi_j^i}}(p_i(x))=
  {G_i}_{\overline{\varphi_j^i}}(p_i(y))$,
  i.e., there exists a path $P_{p_i(x),p_i(y)}:=(e_1,\ldots,e_k)$ from
  $p_i(x)$ to $p_i(y)$ in $G_i$ such that $e_l\notin\varphi_j^i$ for $1\leq
  l\leq k$. Let $w\in V(G)$ s.t. $p_i(w)=p_i(y)$ and $p_r(w)=p_r(x)$ for
  all $r\in I$, $r\neq i$. Hence, $w\in V(G^x_i)$. Thus, there exists a
  path $P'_{x,w}=(e'_1,\ldots,e'_k)$ in $G$ with $p_i(e_l')=e_l$ which
  implies $e'_l\notin\psi^i_j$ for $1\leq l\leq k$ and thus
  $G_{\overline{\psi_j^i}}(x)=G_{\overline{\psi_j^i}}(w)$. Furthermore, by
  the properties of the Cartesian product, there exists a path
  $P''_{w,y}=(e''_1,\ldots,e''_s)$ from $w$ to $y$ in $G$ such that
  $|p_i(e''_l)|=1$ for $1\leq l\leq s$, which implies $I_{e''_l}\neq
  \{i\}$ and consequently $e''_l\notin\psi^i_j$ for $1\leq l\leq k$. Thus,
  $G_{\overline{\psi_j^i}}(y)=G_{\overline{\psi_j^i}}(w)=
  G_{\overline{\psi_j^i}}(x)$,
  that is, this mapping is injective and therefore bijective. It remains to
  show that $[G_{\overline{\psi_j^i}}(x),G_{\overline{\psi_j^i}}(y)]$ is an
  edge in $G_{\psi^i_j}/\mc P^R_{\overline{\psi^i_j}}$ if and only if
  $[{G_i}_{\overline{\varphi_j^i}}(p_i(x)),
  {G_i}_{\overline{\varphi_j^i}}(p_i(y))]$
  is an edge in ${G_i}_{\varphi^i_j}/\mc P^{R_i}_{\overline{\varphi^i_j}}$.
  By definition,
  $[G_{\overline{\psi_j^i}}(x),G_{\overline{\psi_j^i}}(y)]$ is an edge in
  $G_{\psi^i_j}/\mc P^R_{\overline{\psi^i_j}}$ if and only if there exists
  $x'\in V(G_{\overline{\psi_j^i}}(x)), y'\in
  V(G_{\overline{\psi_j^i}}(x))$ s.t. $[x',y']\in\psi^i_j$, which, by the
  preceding and by construction, is equivalent to $p_i(x')\in
  V({G_i}_{\overline{\varphi_j^i}}(p_i(x))), p_i(y')\in
  V({G_i}_{\overline{\varphi_j^i}}(p_i(y)))$ and $[p_i(x'),p_i(y')]\in
  \varphi^i_j$, from what the assertion follows.

  \bigskip $(\boxtimes)$ 
  To prove the assertion, we have to show that the spanning subgraph
  $G_{\overline\varphi}$ is connected for all $\varphi\sqsubseteq R$. For
  each $\varphi\sqsubseteq R$ it holds that $I_e=I_f$ for all
  $e,f\in\varphi$. We set $I_\varphi:=I_e$ for some $e\in\varphi.$
  Moreover, define $\Phi:=\{\psi\sqsubseteq R\mid I_\psi=I_\varphi\}$ Then
  for $\alpha:=\bigcup_{\psi\in\Phi}\psi$, $G_{\overline\alpha}$ is a
  spanning subgraph of $G_{\overline\varphi}$. Therefore, it suffices to
  show that $G_{\overline\alpha}$ is connected. To be more precise, we have
  to show that for all $x,y\in V(G)$, there exists a walk $W_{x,y}$ from
  $x$ to $y$ in $G$ such that for all $e\in E(W_{x,y})$ it holds that $I_e\neq
  I_\varphi$.

  First, assume $|I_\varphi|>1$.  Since $\Box_{i\in I} G_i$ is a connected
  spanning subgraph of $\boxtimes_{i\in I} G_i$, there exists a walk
  $W_{x,y}$ from $x$ to $y$ in $\Box_{i\in I} G_i$. Then for all $e\in
  E(W_{x,y})$ it holds that $|I_e|=1$ and thus, $I_e\neq I_\varphi$.

  Now, let $|I_\varphi|=1$, i.e., $I_\varphi=\{j\}$ for some $j\in I$.  If
  $p_j(x)=p_j(y)$, then $y\in V((\Box_{i\in I\setminus\{j\}} G_i)^x)$.  In
  this case, there exists a walk $W_{x,y}$ from $x$ to $y$ in 
  $(\Box_{i\in I\setminus\{j\}} G_i)^x$ that has the desired properties. If
  $p_j(x)\neq p_j(y)$, let $y'\in V(G)$ such that $p_i(y')=p_i(x)$ for all
  $i\neq j$ and $p_j(y)=p_j(y')$. Then, as in the previous case, there
  exists a walk $W_{y,y'}$ from $y$ to $y'$ in $(\Box_{i\in
    I\setminus\{j\}} G_i)^y$ and hence $I_e\neq\{j\}$ for all $e\in
  E(W_{y,y'})$. By choice of $y'$, it holds that $y'\in V(G_j^x)$. Let
  $P_{x,y'}:=(x=x_0,x_1,\ldots,x_k=y')$ be a walk from $x$ to $y'$ that is
  entirely contained in $G_i^x$. Moreover, for arbitrary $i\in I$ with
  $i\neq j$ let $z\in V(G_i^x)$ such that $[p_i(x),p_i(z)]\in E(G_i)$ and
  let $w\in V(G_j^z)$ such that $p_j(w)=p_j(z)$. Then there exists a walk
  $P_{z,w}:=(z=z_0,z_1,\ldots,z_k=w)$ from $z$ to $w$ in $G_j^z$ such that
  $p_j(x_r)=p_j(z_r)$ for all $0\leq r\leq k$. By definition of
  $\boxtimes$,
  $W_{x,y'}:=(x_0,z_1,x_1,z_2,x_2,z_3,\ldots,x_{k-1},z_k=w,x_k=y')$ is a
  walk from $x$ to $y'$ in $G$ and for the edges $e\in E(W_{x,y'})$ it holds that
  $I_e=\{i,j\}\neq \{j\}=I_\varphi$ if $e$ is of the form $[x_i,z_{i+1}]$,
  $0\leq i\leq k-1$ and $I_e=\{i\}\neq\{j\}=I_\varphi$ if $e$ is of the
  form $[x_i,z_i]$, $0\leq i\leq k$. Hence, $W_{x,y}=W_{x,y'}\cup W_{y',y}$
  is a walk from $x$ to $y$ that has the desired properties.
\end{proof}

In contrast to the Cartesian and strong products, no general statement can
be obtained for the direct product $G=\times_{i\in I}G_i$ of graphs $G_i$
since the structure of direct products strongly depends on additional
properties such as bipartiteness.

\section{RSP-Relations on Complete and Complete Bipartite Graphs}

Since complete graphs and complete bipartite graphs contain large numbers
of superimposed $K_{2,3}$ subgraphs they are responsible for much of the
difficulties in finding finest RSP-relations. We therefore study their
RSP-relations in some detail. 

\begin{lemma}
  Let $V(K_m)=\{0,\ldots,m-1\}$. For $i=1,\ldots,l:=\lfloor
  \frac{m}{2}\rfloor$ define the set \[\varphi_i:=\{[x,(x+i)\mod m]\mid
  x\in\{0,\ldots m-1\}\}\subseteq E(K_m).\] Then the sets
  $\varphi_1,\ldots,\varphi_l$ define an RSP-relation $R$ on $E(K_m)$ with
  equivalence classes $\varphi_1,\ldots,\varphi_l$. If $m\neq 4$, then $R$
  is a finest RSP-relation.
\label{lem:Km}
\end{lemma}

\begin{proof}
  At first we prove that $R$ is an equivalence relation. That is, we have
  to show that $\varphi_i\cap\varphi_j=\emptyset$ for all $i\neq j$ and
  $E(K_m)=\bigcup_{i=1}^l \varphi_i$. For contraposition suppose,
  $\varphi_i\cap\varphi_j\neq \emptyset$ for some $i\neq j$. That is, there
  exists $x,y\in V(K_m)=\{0\ldots,m-1\}$ such that $[x,(x+i)\mod
  m]=[y,(y+j)\mod m]$. Notice, $x+i< 2m$ as well as $y+j< 2m$.
  Thus, we have $x+i=p\cdot m+(x+i)\mod m$ and $y+j=q\cdot m+(y+j)\mod m$
  with $p,q\in\{0,1\}$. First assume $x=y$. Hence, $(x+i)\mod m=(x+j)\mod
  m$ and we obtain $|i-j|=|p-q|\cdot m$ with $|q-p|\in\{0,1\}$. If
  $|p-q|=0$ it follows $i=j$. Therefore suppose, $|p-q|=1$. This implies
  $|i-j|=m\geq 2l$ and moreover,  $|i-j|<l$ since
  $i,j\in\{1,\ldots,l\}$, a contradiction.

  Now, assume $x\neq y$. Then it must hold $x=(y+j)\mod m$ and $y=(x+i)\mod
  m$ if $[x,(x+i)\mod m]=[y,(y+j)\mod m]$. Hence, with our considerations
  above, we get $i+j=(p+q)\cdot m$ with $p+q\in\{0,1,2\}$.  from
  $i,j\in\{1,\ldots,l\}$, we conclude $0<i+j\leq 2l$ which implies in
  particular $p+q>0$. It follows $2l\leq m\leq i+j\leq 2l$, hence $i=j=l$
  which contradicts the choice of $i,j$. Thus,
  $\varphi_i\cap\varphi_j=\emptyset$ for all $i,j\in\{1,\ldots,l\}$
  with $i\neq j$.

  Next, we show $|\bigcup_{i=1}^l \varphi_i| =|E(K_m)|$. Since
  $\varphi_i\subseteq E(K_m)$ for all $i\in\{1,\ldots,l\}$, we then can
  conclude $\bigcup_{i=1}^l \varphi_i=E(K_m)$. First, let $i<\frac{m}{2}$.
  Assume, there exists $x\in \{0,\ldots,m-1\}$ such that $x=(x+i)\mod m$.
  From previous considerations, it follows $i=p\cdot m$ with $p\in\{0,1\}$,
  which contradicts $0<1\leq i\leq l<m$. Now suppose, there are
  $x,y\in\{0,\ldots,m-1\}$ such that $[x,(x+i)\mod m]=[y,(y+i)\mod m]$. if
  $x\neq y$, it follows $x= (y+i)\mod m$ and $y=(x+i)\mod m$. As before, we
  conclude $2i=(p+q)\cdot m$ with $p+q\in\{0,1,2\}$ and since $i>0$, we
  have $p+q>0$. Thus, $m\leq 2i<m$, which is a contradiction. Hence,
  $|\varphi_i|=|\{0,\ldots,m-1\}|=m$ for all $i<\frac{m}{2}$. If
  $i=\frac{m}{2}$, and thus, $m$ is even, we have
  $|\varphi_{\frac{m}{2}}|=\frac{m}{2}$, since for all $x<\frac{m}{2}$
  it holds that $[x,x+\frac{m}{2}]=[x+\frac{m}{2},(x+\frac{m}{2}+\frac{m}{2})\mod
  m]$.  It follows $|\bigcup_{i=1}^l
  \varphi_i|=\sum_{i=1}^l|\varphi_i|=l\cdot m=\frac{(m-1)\cdot m}{2}
  =|E(K_m)|$ if $m$ is odd and $|\bigcup_{i=1}^l
  \varphi_i|=\sum_{i=1}^{l-1}|\varphi_i|+|\varphi_{\frac{m}{2}}|=(l-1)\cdot
  m+\frac{m}{2}=\frac{(m-1)\cdot m}{2} =|E(K_m)|$ if $m$ is even.
  Therefore, $R$ is an equivalence relation on $E(K_m)$.
    
  It remains to show that $R$ has the relaxed square property and there is
  no refinement of $R$ with this property. Therefore, let
  $e=[x,y]\in\varphi_i$ and $f=[x,z]\in\varphi_j$, $i\neq j$. We have to
  show, that there exists a vertex $w\in V(K_m)$ such that
  $[y,w]\in\varphi_j$ and $[z,w]\in\varphi_i$.  $[x,y]\in\varphi_i$ implies
  $y=(x+i)\mod m$ or $x=(y+i)\mod m$ and $[x,z]\in\varphi_j$ implies
  $z=(x+j)\mod m$ or $x=(z+j)\mod m$.  If $y=(x+i)\mod m$ and $z=(x+j)\mod
  m$, we choose $w=(y+j)\mod m$. It is clear, that $w\neq x,y,z$. By
  definition, it holds that $[y,w]\in\varphi_j$.  Moreover, by simple
  calculation we get with the preceding $w=(z+i)\mod m$ and hence
  $[z,w]\in\varphi_i$. If $y=(x+i)\mod m$ and $x=(z+j)\mod m$, we choose
  $w=(z+i)\mod m$, then $w\neq x,y,z$. Hence, $[w,z]\in\varphi_i$. In this
  case we get $y=(w+j)\mod m$ that is $[y,w]\in\varphi_j$. If $x=(y+i)\mod
  m$ and $z=(x+j)\mod m$, we choose $w=(y+j)\mod m$. Again  $w\neq
  x,y,z$. and by definition, $[y,w]\in\varphi_j$. Here, we obtain
  $z=(w+i)\mod m$ and hence $[z,w]\in\varphi_i$. If $x=(y+i)\mod m$ and
  $x=(z+j)\mod m$, we choose $w$ such that $z=(w+i)\mod m$, that is
  $[z,w]\in\varphi_i$. In this case we have $w\neq x,y,z$ and moreover,
  $y=(w+j)\mod m$ and hence $[y,w]\in\varphi_j$. That is, $R$ has the
  relaxed square property.

  We show now, that no equivalence class $\varphi$ of $R$ can be split
  into two classes $\varphi_i=\psi_{i_1}\cup\psi_{i_2}$, such that the
  equivalence relation, $S$ that has classes
  $\varphi_1,\ldots,\varphi_{i-1},\psi_{i_1},\psi_{i_2},\varphi_{i+1},
  \ldots,\varphi_l$ is an RSP-relation. Therefore, notice that each vertex
  $x\in V(K_m)$ is incident to exactly two $\varphi_i$ edges for all
  $i<\frac{m}{2}$, namely $[x,(x+i) \mod m]$ and $[x,(x-i) \mod m]$, thus
  the layer are all cycles for $i<\frac{m}{2}$. Moreover, each vertex $x\in
  V(K_m)$ is incident to exactly one
  $\varphi_{\frac{m}{2}}$-edge. Recalling Lemma~\ref{lem:classes},
  $\varphi_{\frac{m}{2}}$ cannot be split. For $k<\frac{m}{2}$ let $C$
  the $\varphi_k$-layer containing vertex $0$. It has edges
  $[0,k],[k,2k],[2k,3k\mod m],\ldots,[(q-1)\cdot k,0]$ with $q\cdot k\mod
  m=0$. By Lemma~\ref{lem:delete_classes}, any edge in $C$ must be
  contained in a square, hence $C$ itself must be a square and thus has
  edges $[0,k],[k,2k],[2k,3k],[3k,0]$ with $4k=m$, since $k<\frac{m}{2}$
  and $k>1$ since $m\neq 4$. Because $S$ is an RSP-relation, it holds that
  $([0,k],[2k,3k]),([k,2k],[3k,0])\in S$ and
  $([0,k],[k,2k]),([2k,3k],[3k,0])\notin S$ by Lemma~\ref{lem:classes}.
  Consider the edges $[0,k]\in\varphi_k$ and
  $[0,1]\in\varphi_1\neq\varphi_k$, hence they are in different
  $S$-classes.  Vertex $k\in V(K_m)$ is incident to exactly two
  $\varphi_1$-edges, namely $[k,k+1]$ and $[k,k-1]$. Since
  $[1,k-1]\in\varphi_{k-2}\neq\varphi_k$, the only possible square spanned
  by $[0,k]$ and $[0,1]$ with opposite edges in the same $S$-class is
  $0-1-(k+1)-k$ with $[0,k],[k,k+1]\in S$. Now, consider edges
  $[k,2k]\in\varphi_k$ and $[1,k]\in\varphi_{k-1}$. Vertex $2k\in V(K_m)$
  is incident to exactly two $\varphi_{k-1}$-edges, namely $[2k,k+1]$ and
  $[2k,3k-1]$. Since $[1,3k-1]\in\varphi_{k+2}\neq\varphi_k$, the only
  possible square spanned by $[k,2k]$ and $[1,k]$ with opposite edges in
  the same $S$-class is $1-k-2k-(k+1)$ with $([1,k+1],[k,2k])\in
  S$. Thus, $([0,k],[k,2k])\in S$, a contradiction. Hence, $R$ is finest
  RSP-relation on $K_m$ for all $m\neq 4$.
  \end{proof}

\begin{coro}
  For all $m>3$ there exists a nontrivial RSP-relation on $E(K_m)$.
\end{coro}

Lemma~\ref{lem:Km} implies that the maximal number of classes of a finest
RSP-relation is at least $\lfloor\frac{m}{2}\rfloor$.  From
Lemma~\ref{lem:classes}, we infer that the maximal number of classes of a
finest RSP-relation on $K_m$ is at most $m-1$, the minimum degree of $K_m$.
In the case of $m=2^q$, this bound is sharp with the construction in
Definition~\ref{def:rel_prod} and since $K_{2^q}= \boxtimes_{i=1}^q K_2$.

To show the large variety of possible finest RSP-relations on complete graphs
we give a further example. 
\begin{example}
For $n\geq 5$ and graph $K_n$, let $G_1$ be the induced subgraph on
vertices $\{ 0, 1 \}$ and $G_2$ the induced subgraph on $\{ 2,\ldots,
n-1\}$. We claim that relation $R$ with two equivalence classes
$\varphi=E(G_1) \cup E(G_2)$ and $\overline{\varphi}$ is a finest RSP
relation. It is easy to check that it is an RSP-relation. Equivalence class
$\varphi$ cannot be split into two equivalence classes since vertex $0$ is
incident with only one edge of $\varphi$. On the other hand, every vertex
in $\{ 2,\ldots, n-1\}$ is incident with exactly two edges in
$\overline{\varphi}$, therefore if $\overline{\varphi}$ can be split into
two equivalence classes edges $[0,2]$ and $[1,2]$ must be in different
equivalence classes. The definition of RSP-relations implies that $[0,2]$
and $[2,3]$ must lie on a common square with opposite edges the same equivalence class.
The only possible candidate is the square $0-2-3-1$, thus $[0,2]$ and
$[1,3]$ must be in the same class. Similarly, $[1,3]$ and $[3,4]$ must lie on a
common square with opposite edges in the same equivalence class. The only possible candidate
is the square $1-3-4-0$, thus $[1,3]$ and $[0,4]$ must be in the same class. Now,
we use the same arguments for edges $[0,4]$ and $[4,2]$ to find out that
$[0,4]$ and $[1,2]$ are in the same class. Since the relation is transitive
$[0,2]$ and $[1,2]$ must be in the same class, a contradiction with the
assumption that $\overline{\varphi}$ can split.
\label{exmpl:anotherKM}
\end{example}


\begin{example}
Consider the complete graph
$K_9=K_3\boxtimes K_3$. Then the construction given in Lemma
\ref{lem:Km} and in Lemma
\ref{lem:rel_prod} define two different RSP-relations $R\not\simeq S$,
for which 
$K_9/\mc P^R\simeq K_9/\mc P^S \simeq \mathcal L K_1$,
by Lemma \ref{lem:prod-quot}.  
Note, $R$ and $S$
have no RSP-relation as common refinement.
\label{ex:k9}
\end{example}

Let us now turn to complete bipartite graphs $K_{m,n}$. W.l.o.g.\ we may 
assume that $m\le n$.
\begin{lemma}
  \label{lem:Kmm}
  For $m=n$ let the vertex set of $K_{m,m}$ be given by
  $V(K_{m,m})=V(K_2)\times V(K_m)$ and $E(K_{m,m})=\{[x,y]\mid x,y\in
  V(K_{m,m})\text{ s.t. } p_1(x)\neq p_1(y)\}$.  Furthermore, let $S$ be an
  RSP-relation on $E(K_m)$.  We define an equivalence relation $R$ on
  $E(K_{m,m})$ as follows: $(e,f)\in R$ if and only if
  \begin{itemize}
  \item[(1)] $|p_2(e)|=|p_2(f)|=1$, or
  \item[(2)] $|p_2(e)|=|p_2(f)|=2$ and $(p_2(e),p_2(f))\in S$.
  \end{itemize}
  Then $R$ has the relaxed square property.
  Moreover, $R$ is a finest RSP-relation on $E(K_{m,m})$
  if and only if $S$ is finest RSP-relation on $E(K_m)$.
\end{lemma}
\begin{proof}
  Notice, that with our notation we have
  $E(K_{m,m})=E(K_2\boxtimes K_m)\setminus (E(K_m^x)\cup E(K_m^y))$
  with $x,y\in V(K_2)\times V(K_m)$ s.t. $p_1(x)\neq p_1(y)$.
  With Lemma~\ref{lem:rel_prod} and Lemma~\ref{lem:delete_classes},
  it follows that $R$ is an RSP-relation on $E(K_{m,m})$.
  It is clear that any refinement of $S$ leads to a refinement of $R$.
  Thus we just have to show the converse, i.e., that $R$ is a finest
  RSP-relation if $S$ is finest RSP-relation.
  Let $\varphi$ denote the equivalence class defined by condition (1),
  i.e., $\varphi=\{e\in E(K_{m,m})\mid |p_2(e)|=1\}$.
  By construction, each vertex is adjacent to exactly one $\varphi$-edge,
  therefore, $\varphi$ cannot be split by Lemma~\ref{lem:classes}.
  Moreover, two adjacent edges $e,f$ with $e\in\varphi$
  and $f\in\psi\neq\varphi\eqcl R$ span exactly one square with
  opposite edges in the same equivalence classes,
  namely the square with $p_2(f)=p_2(f')$, where $f'$
  is opposite edge of $f$.
  Therefore, $p_2(e)=p_2(e')$ implies $(e,e')\in Q$ for any refinement
  $Q$ of $R$ with relaxed square property.
    Furthermore, with our notations, any refinement $Q$ of $R$
  leads also to a refinement $Q_{|E(K_2\times K_m)}$ of
  $R_{|E(K_2\times K_m)}$, the restrictions of $Q$ and $R$
  to $E(K_2\times K_m)\subseteq E(K_{m,m})$, respectively.
  If the refinement $Q$ is proper and satisfies the relaxed square property
  on $E(K_{m,m})$, the same is true for $Q_{|E(K_2\times K_m)}$ on 
  $E(K_2\times K_m)$
  by Lemma~\ref{lem:delete_classes} and our previous considerations.
  Moreover, we can conclude that $Q$ determines an equivalence relation
  $p_2(Q)$ on $K_m$ via $(p_2(e),p_2(f))\in p_2(Q)$ iff $(e,f)\in Q$.
  It holds $p_2(C_4)\cong C_4$ for any square in $K_2\times K_m$.
  Furthermore, $p_2(e)=p_2(e')$ implies $(e,e')\in Q$ if $Q$ has the relaxed square 
  property. 
  Therefore, it follows, $p_2(Q)$ is a proper refinement of $S$ with 
  the relaxed square property if $Q$ is a proper refinement of $R$ with 
  the relaxed square property.  This completes the proof.  
\end{proof}

\begin{lemma}
  \label{lem:Kmn}
  For $m<n$ let the vertex set of $K_{m,n}$ be given by
  $\{x_1,\ldots,x_m,y_1,\ldots,y_n\}$ such that $E(K_{m,n})=\{[x_i,y_j]\mid
  1\leq i\leq m, 1\leq j\leq m\}$. Furthermore,let $S$ be an equivalence
  relation on the edge set of the induced subgraph
  $\langle\{x_1,\ldots,x_m,y_1,\ldots y_m\}\rangle\cong K_{m,m}$ of
  $K_{m,n}$. We extend $S$ to an equivalence relation $R$ on $E(K_{m,n})$ as
  follows: For each equivalence class $\varphi'\sqsubseteq S$ we extend
  $\varphi'$ to an equivalence class $\varphi\sqsubseteq R$, i.e., we set
  $\varphi'\subseteq\varphi$ and moreover $[x_j,y_{m+i}]$ is an edge in
  equivalence class $\varphi$ if and only if $[x_j,y_{k_i}]$ is an edge in
  $\varphi'$ for fixed $k_i\in\{1,\ldots,m\}$ for all
  $i\in\{1,\ldots,n-m\}$. Then $R$ has the relaxed square property. 
\end{lemma}
\begin{proof}
  It is clear, that $R$ is an equivalence relation. Thus, it remains to
  show that $R$ has the relaxed square property. Therefore, let $e,f\in
  E(K_{m,n})$ such that $(e,f)\notin R$. Notice, by construction it holds that
  $\psi'\neq\varphi'$ if and only if $\psi\neq\varphi$ for all
  $\psi',\varphi'\sqsubseteq S$ and $\psi,\varphi\sqsubseteq R$ with
  $\psi'\subseteq\psi$ and $\varphi'\subseteq\varphi$.
  
  First, suppose that $e$ and $f$ are incident in some vertex $y_r\in
  V(K_{m,n})$, $r\in\{1,\ldots,n\}$. That is, $e=[x_j,y_r]$ and
  $f=[x_l,y_r]$ for some $j,l\in\{1,\ldots,m\}, j\neq l$. If $r\leq m$ then
  by construction $e,f\in E(K_{m,m})$ and $(e,f)\notin S$, and hence they
  span a square with opposite edges in the same equivalence classes of $S$,
  which is also retained in $K_{m,n}$ with the same properties. If $r>m$,
  then $r=m+i$ for some $i\in\{1,\ldots,n-m\}$. By construction, there
  exists $k_i\in\{1,\ldots,m\}$ such that $([x_j,y_{k_i}],[x_j,y_{m+i}])\in
  R$ and $([x_l,y_{k_i}],[x_l,y_{m+i}])\in R$, which implies
  $([x_j,y_{k_i}],[x_l,y_{k_i}])\notin R$ and hence, by construction,
  $([x_j,y_{k_i}],[x_l,y_{k_i}])\notin S$. Since $S$ has the relaxed square
  property, there exists $w\in V(K_{m,m})\subset V(K_{m,n})$ such that
  $[x_j,y_{k_i}]$ and $[x_l,y_{k_i}]$ span a square $x_j-y_{k_i}-x_l-w$,
  such that $([x_l,w],[x_j,y_{k_i}])\in S\subset R$ and
  $([x_j,w],[x_l,y_{k_i}])\in S\subset R$. Then $x_j-y_{m+i}-x_l-w$ is a
  square spanned by $e$ and $f$ with opposite edges in the same equivalence
  class.
    
  Now assume $e$ and $f$ are incident in some vertex $x_j\in V(K_{m,n})$,
  $j\in\{1,\ldots,m\}$. That is, $e=[x_j,y_r]$ and $f=[x_j,y_s]$ for some
  $r,s\in\{1,\ldots,n\}, r\neq s$. If $r,s\leq m$, then by construction
  $e,f\in E(K_{m,m})$ and $(e,f)\notin S$, and hence they span a square with
  opposite edges in the same equivalence classes of $S$, which is also
  retained in $K_{m,n}$ with the same properties. If $r,s>m$, then $r=m+i$,
  $s=m+l$ for some $i,l\in\{1,\ldots,n-m\}$. By construction, there exists
  $k_i,k_l\in\{1,\ldots,m\}$ such that $([x_j,y_{m+i}],[x_j,y_{k_i}])\in R$
  as well as $([x_j,y_{m+l}],[x_j,y_{k_l}])\in R$, from which we can
  conclude $([x_j,y_{k_i}],[x_j,y_{k_l}])\notin R$. By construction we have
  $([x_j,y_{k_i}],[x_j,y_{k_l}])\notin S$, and since $S$ has the relaxed
  square property, there exists $w\in V(K_{m,m})\subset V(K_{m,n})$ such
  that $[x_j,y_{k_i}]$ and $[x_j,y_{k_l}]$ span a square
  $(x_j,y_{k_i},w,y_{k_l})$, such that $([w,y_{k_l}],[x_j,y_{k_i}])\in
  S\subset R$ and $([w,y_{k_i}],[x_j,y_{k_l}])\in S\subset R$. Moreover, by
  construction, we have $([w,y_{m+i}],[w,y_{k_i}])\in R$ as well as
  $([w,y_{m+l}],[w,y_{k_l}])\in R$. Thus $x_j-y_{m+i}-w-y_{m+l}$ is a
  square spanned by $e$ and $f$ with opposite edges in the same equivalence
  class. If $r>m,s\leq m$, then $r=m+i$ for some $i\in\{1,\ldots,n-m\}$. By
  construction, there exists $k_i\in\{1,\ldots,m\}$ such that
  $([x_j,y_{m+i}],[x_j,y_{k_i}])\in R$ and thus,
  $([x_j,y_{k_i}],[x_j,y_{l}])\notin R$, hence,
  $([x_j,y_{k_i}],[x_j,y_{k_l}])\notin S$. Since $S$ has the relaxed square
  property, there exists $w\in V(K_{m,m})\subset V(K_{m,n})$ such that
  $[x_j,y_{k_i}]$ and $[x_j,y_{l}]$ span a square $x_j-y_{k_i}-w-y_{l}$,
  such that $([w,y_{l}],[x_j,y_{k_i}])\in S\subset R$ and
  $([w,y_{k_i}],[x_j,y_{l}])\in S\subset R$. Moreover, by construction, we
  have $([w,y_{m+i}],[w,y_{k_i}])\in R$. Hence, $x_j-y_{m+i}-w-y_{l}$ is a
  square spanned by $e$ and $f$ with opposite edges in the same equivalence
  class. Analogously, one shows that $e$ and $f$ span a square with
  opposite edges in the same equivalence class if $r\leq m$ and $s>m$, which
  completes the proof.
\end{proof}

Obviously, any finer RSP-relation $S'\subset S$ on $E(K_{m,m})$
leads to a finer RSP-relation $R'\subset R$ on $E(K_{m,n})$,
constructed from $S'$ as in Lemma~\ref{lem:Kmn}.
It is not known yet, if the converse is also true.

\begin{coro}
  For all $m,n\geq 2$ there exists a nontrivial RSP-relation on $E(K_{m,n})$.
\end{coro}

The constructions in Lemma~\ref{lem:Kmm} and Lemma~\ref{lem:Kmn} together
with Lemma~\ref{lem:Km} imply that the maximal number of classes of a
finest RSP-relation is at least $\lfloor\frac{m}{2}\rfloor+1$.  From
Lemma~\ref{lem:classes}, we infer that the maximal number of classes of a
finest RSP-relation on $K_{m,n}$ is at most $m$, the minimum degree of
$K_{m,n}$.  In the case of $m=2^q$, this bound is sharp with our
considerations for complete graphs $K_{2^q}$ and the constructions in
Lemma~\ref{lem:Kmm} and Lemma~\ref{lem:Kmn}.

\section{RSP-relations and Covering Graphs}

We are now in the position, to establish the close connection of covering
graphs and (well-behaved) RSP-relations.

\begin{defi}
  For a graph $G=(V,E)$, an RSP-relation $R$ on $E$ and $\varphi \eqcl R$,
  let $G_{\varphi}^x$ and $G_{\varphi}^y$ be two distinct adjacent
  $\varphi$-layer. We define the graph $C_{G_{\varphi}^x,G_{\varphi}^y}$ in
  the following way:
\begin{enumerate}
\item Vertices $V(C_{G_{\varphi}^x, G_{\varphi}^y}) = \{[a,b]\in E\mid
  a\in V(G_{\varphi}^x), b\in V(G_{\varphi}^y) \}$ are precisely the edges
  of $G$ connecting $G_{\varphi}^x$ and $G_{\varphi}^y$.
\item Two vertices $[a_1,b_1],[a_2,b_2]\in
  V(C_{G_{\varphi}^x,G_{\varphi}^y})$ are adjacent if they are opposite
  edges of a square $a_1-b_1-b_2-a_2$ in $G$ with $[a_1,a_2]\in
  E(G_{\varphi}^x)$ and $[b_1,b_2] \in E(G_{\varphi}^y)$.
\end{enumerate}
\end{defi}


\begin{lemma}\label{lem:cover}
  Let $G$ be a graph, $R$ an RSP-relation on $E(G)$, and
  $G_{\varphi}^x$ and $G_{\varphi}^y$ two distinct adjacent $\varphi$-layer
  for some $\varphi \eqcl R$.  Then $C_{G_{\varphi}^x,G_{\varphi}^y}$ is a
  quasi-cover of $G_{\varphi}^x$ and $G_{\varphi}^y$.  Moreover, if $R$ is
  well-behaved, then $C_{G_{\varphi}^x,G_{\varphi}^y}$ is a cover
  of $G_{\varphi}^x$ and $G_{\varphi}^y$.
\end{lemma}

\begin{proof}
  We define the map $f_1: V(C_{G_{\varphi}^x,G_{\varphi}^y}) \rightarrow
  V(G_{\varphi}^x)$ by $f_1([a,b])=a$ where $a \in V(G_{\varphi}^x)$ and $b
  \in V(G_{\varphi}^y)$ and show first that $f_1$ is a homomorphism, i.e.,
  it maps neighbors in $C_{G_{\varphi}^x,G_{\varphi}^y} $ into neighbors in
  $G_{\varphi}^x$.  Let $[a_1,b_1], [a_2,b_2] \in
  V(C_{G_{\varphi}^x,G_{\varphi}^y})$ be adjacent.  By construction of
  edges in $C_{G_{\varphi}^x,G_{\varphi}^y}$, there is a square
  $a_1-b_1-b_2-a_2$ in $G$ with opposite edges $[a_1,b_1]$ and $[a_2,b_2]$.
  Hence, $a_1$ and $a_2$ are adjacent in $G_{\varphi}^x$.  Now, let
  $a=f_1([a,b])$ and $c \in N_{G_{\varphi}^x }(a)$. Since $[a,c]$ and
  $[a,b]$ are incident edges of different equivalence classes, they span
  some square with opposite edges in relation $R$.  Thus there exists a
  vertex $d\in V(G_{\varphi}^y)$, such that $[a,b]$ and $[c,d]$ are
  adjacent in $C_{G_{\varphi}^x,G_{\varphi}^y}$ and $f_1([c,d])=c$. This
  proves that $f_1$ is locally surjective and therefore, that
  $C_{G_{\varphi}^x,G_{\varphi}^y}$ is a quasi cover of $G_{\varphi}^x$.

  Let $f_1$ be defined as above and assume that none of the subgraphs of
  $G$ that are isomorphic to $K_{2,3}$ have a forbidden coloring.
  If $f_1([c_1,d_1])=f_1([c_2,d_2])$ it holds that for $[c_1,d_1],[c_2,d_2] \in
  N_{C_{G_{\varphi}^x,G_{\varphi}^y}}([a,b])$ we have $c_1=c_2$ by 
  construction of $f_1$. If $d_1\neq d_2$, then there is a subgraph of $G$
  isomorphic to $K_{2,3}$ with bipartition $\{b,c_1\}\dot\cup
  \{a,d_1,d_2\}$. Moreover, since $[a, c_1], [b, d_1], [b, d_2] \in
  \varphi$ and the other edges are, by construction, in $\vpo$ we conclude
  that this subgraph has a forbidden coloring, a contradiction. Thus,
  $d_1=d_2$, i.e., the locally surjective map $f_1$ is also locally
  injective. Hence, $C_{G_{\varphi}^x,G_{\varphi}^y}$ is a cover of
  $G_{\varphi}^x$.

  Arguing analogously for the map $f_2: V(C_{G_{\varphi}^x,G_{\varphi}^y})
  \rightarrow V(G_{\varphi}^y)$ with $f_2([a,b])=b$, $a \in
  V(G_{\varphi}^x)$, $b \in V(G_{\varphi}^y)$, one obtains the desired
  results for $C_{G_{\varphi}^x,G_{\varphi}^y}$ and $G_{\varphi}^y$.
\end{proof}

To illustrate Lemma~\ref{lem:cover} consider the following example: Let
$G_1=C_6$ and $G_2=C_9$ with vertex sets $\mathbb{Z}_6$ and $\mathbb{Z}_9$
and the canonical edge set definitions. To obtain $G$ add the edges $[k,k
\mod 6]$ and $[k,k+3 \mod 6]$ for $0\leq k \leq 9$ connecting $G_1$ with
$G_2$. Construct an equivalence relation $R$ with two classes $\varphi=
E(G_1) \cup E(G_2)$, and $\overline{\varphi}$ comprising the connecting
edges. $R$ is a well-behaved RSP-relation on $G$. 
It is not hard to verify that $C_{G_1,G_2}$ is a cover graph of
$C_6$ and $C_9$ and is isomorphic to $C_{18}$.

For a similar result for the case when $G_{\varphi}^x$ and $G_{\varphi}^y$
are not distinct, that is $G_{\varphi}^x=G_{\varphi}^y$, but there are
edges not in $\varphi$ connecting its vertices, we have to be a bit more
careful.

\begin{defi}\label{def:superimposed}
  For a graph $G=(V,E)$, an RSP-relation $R$ on $E$, and $\varphi \eqcl R$,
  let $G_{\varphi}^x$ be some $\varphi$-layer.  We define the graph
  $C_{G_{\varphi}^x,G_{\varphi}^x}$ in the following way:
\begin{enumerate}
\item Vertices $V( C_{G_{\varphi}^x, G_{\varphi}^x}) = \{(a,b)\mid [a,b]\in
  E, a,b\in V(G_{\varphi}^x), [a,b]\in \vpo, \varphi\eqcl R\}$ are edges in $E(G)$ with superimposed 
  orientation $(a,b)$ from $a$ to $b$, that are not contained in class
  $\varphi$, but that connect vertices of $G_{\varphi}^x$.
\item Two directed edges $(a_1,b_1)$ and $(a_2,b_2)$ in
  $V(C_{G_{\varphi}^x,G_{\varphi}^x})$ are adjacent if $[a_1,b_1]$,
  $[a_2,b_2]$ are opposite edges of a square $a_1-b_1-b_2-a_2$ in $G$ with
  $[a_1,a_2],[b_1,b_2]\in E(G_{\varphi}^x)$.
\end{enumerate}
\end{defi}

\begin{remark}
Since $[a,b]=[b,a]$, it  holds that for all edges $[a,b]\in E$,
we get two vertices in $V(C_{G_\vp^x,G_\vp^x})$ per edge $[a,b]\in E\setminus\varphi$, 
namely $(a,b)$ and $(b,a)$.
\end{remark}

\begin{lemma}\label{lem:cover2}
  For a graph $G=(V,E)$, an RSP-relation $R$ on $E$, and $\varphi \eqcl R$,
  let $G_{\varphi}^x$ be some $\varphi$-layer and assume that there are edges
  $[a,b]\in E\setminus\varphi$ with $a,b\in V(G_{\varphi}^x)$.  Then
  $C_{G_{\varphi}^x,G_{\varphi}^x}$ is a quasi-cover of $G_{\varphi}^x$
  with two different locally surjective homomorphisms $f_1$ and $f_2$ such
  that $f_1(h)\neq f_2(h)$ for every $h\in
  C_{G_{\varphi}^x,G_{\varphi}^x}$. Moreover, if $R$ is well-behaved, 
  then $C_{G_{\varphi}^x,G_{\varphi}^x}$ is twice a cover of
  $G_{\varphi}^x$, i.e., there are at least two different covering 
  maps.
\end{lemma}
\begin{proof}
Proof is the same as for Lemma~\ref{lem:cover} by defining $f_1((a,b))=a$ and $f_2((a,b))=b$.
\end{proof}

If every vertex of $G_{\varphi}^x$ is incident with exactly one edge
that is not in $\varphi$ but connects two vertices of $G_{\varphi}^x$, then
$G_{\varphi}^x \cong C_{G_{\varphi}^x,G_{\varphi}^x}$ and the edges in $\vpo$ induce
an automorphism of $G_{\varphi}^x$ without fixed vertices by setting
$f(a)=b$ whenever $[a,b]\in\vpo$.

As an example consider the graph $G$ with
$V(G)=\mathbb Z_6$ and $E(G)=\varphi\dot\cup\vpo$ such that
$\varphi=\{[k,k+1 \mod 6 ]\mid 0\leq k\leq 5\}$, i.e., $G_\varphi\cong C_6$
and $\vpo=\{[1,4],[2,5],[3,6]\}$.  We then have
$V(C_{G_\varphi^x,G_\varphi^x})=\{(0,3),(1,4),(2,5),(3,0),(4,1),(5,2)\}$
and $C_{G_\varphi^x,G_\varphi^x}$ has edges
$E(C_{G_\varphi^x,G_\varphi^x})=\{[(0,3),(1,4)],[(1,4),(2,5)],[(2,5),(3,0)],
[(3,0),(4,1)],[(4,1),(5,2)],[(5,2),(0,2)]\}$,
that is $C_{G_\varphi^x,G_\varphi^x}\cong C_6\cong G_\varphi$.
The induced automorphism is given by $f(k)=k+3\mod 6$, $k=0,\ldots,5$.

Lemma \ref{lem:cover} and Lemma \ref{lem:cover2} together highlight a
connection between graph bundles and graphs with relaxed square property.
For an RSP-relation $R$ on $G$ we see that the connected components
$G_{\varphi}$ correspond to fibers, while the graph
$G_{\overline{\varphi}}/\mc{P}^R_{\varphi}$ has the role of the base graph.
Such decomposition is a graph bundle if and only if edges connecting
$G_{\varphi}^x$ and $G_{\varphi}^y$ for arbitrary $x$, $y$ induce an
isomorphism. In our language, this is equivalent to the condition
$C_{G_{\varphi}^x,G_{\varphi}^y}\cong G_{\varphi}^x \cong G_{\varphi}^y$
for arbitrary $x,\ y$, provided that $G_\vp^x$ and $G_\vp^y$ are connected by an edge.
Graphs with a nontrivial RSP-relation are therefore a natural generalization of
graph bundles.

\begin{coro}
     \label{cor:inters}
  For a graph $G$ and a well-behaved RSP-relation $R$ on $E(G)$, 
  let $G_{\varphi}^x$ and $G_{\varphi}^y$ be two 
  (not necessarily distinct) $\varphi$-layers. 
  Then 
  \begin{equation}
     \label{eq:inters}
    |N_{G_{\overline{\varphi}}}(x)\cap V(G_{\varphi}^y)| = 
    |N_{G_{\overline{\varphi}}}(u)\cap V(G_{\varphi}^y)| 
  \end{equation}
  is fulfilled for every $u\in V(G_{\varphi}^x)$.
\end{coro}
\begin{proof}
  If there is no edge in $G_{\overline{\varphi}}$ connecting
  $G_{\varphi}^x$ and $G_{\varphi}^y$ the assertion is clearly true.
  Therefore assume now that they are connected. By Lemmas~\ref{lem:cover}
  and \ref{lem:cover2}, $C_{G_{\varphi}^x,G_{\varphi}^y}$ is a cover of
  $G_{\varphi}^x$ with covering map $f_1$ as defined in
  Lemmas~\ref{lem:cover} resp. \ref{lem:cover2}.  By definition of $f_1$,
  $|f_1^{-1}(u)|=|N_{G_{\overline{\varphi}}}(u)\cap V(G_{\varphi}^y)|$,
  which is the same for all $u\in V(G_\varphi^x)$.
\end{proof}

Corollary~\ref{cor:inters} indicates another property of well-behaved
RSP-relations. It was shown in \cite{OstermeierL:14} that for a so-called
USP-relation $R$ on $E(G)$ the vertex partitions $P^R_{\vpo}$ and $P^R$
induced by equivalence classes $\vp\eqcl R$ are equitable partitions for
the graphs $G_\vp$ and $G$, respectively.  The key argument leading to this
result was an analogue of Equation~\eqref{eq:inters}.  Together with
Lemma~\ref{lem:delete_classes}, the fact that if $R$ is well-behaved on $G$
then $R\setminus\vp$ is well-behaved on $(V(G),E(G)\setminus\vp)$, and
since $|\dot\bigcup_\psi N_\psi(x)|=\sum_\psi|N_\psi(x)|$ for any set of
pairwisely distinct equivalence classes $\psi$ of $R$, we can use the same
arguments as in \cite{OstermeierL:14} to obtain
\begin{theorem} \label{thm:equitpart} Let $R$ be (a coarsening of) a
  well-behaved RSP-relation on the edge set $E(G)$ of a connected graph
  $G$.  Then:
  \begin{itemize}
  \item[(1)] $\mathcal{P}^R_{\vpo}=\left\{V(G_{\vpo}^x)\mid x\in
      V(G)\right\}$ is an equitable partition of the graph $G_\vp$ for
    every equivalence class $\vp$ of $R$.
  \item[(2)] $\mc P^R=\left\{\bigcap_{\vp\eqcl R}V(G_{\vpo}(x))\mid x\in
      V(G)\right\}$ is an equitable partition of $G$.
  \end{itemize}
\end{theorem}

As mentioned previously, while an RSP-relation $R$ on $E(G)$ might be
well-behaved and thus, has no forbidden $K_{2,3}$-coloring this is no
longer true for coarsenings of $R$ in general.  However, since the number
of edges incident to a vertex is additive over equivalence classes of $R$,
the latter theorem remains also true for coarsenings of relations
\emph{without} forbidden $K_{2,3}$-colorings.

Another interesting question is how two graphs $G_1$ and $G_2$ can be
connected by additional edges so that $\varphi=E(G_1) \cup E(G_2)$ and
$\overline{\varphi}$ comprises the connecting edges and $R=\{\vp,\vpo\}$ is
an RSP-relation.
\begin{lemma}\label{lem:cover3}
  Let $G_1$, $G_2$, and $G$ be graphs and $f_1:G\rightarrow G_1$, $f_2:G
  \rightarrow G_2$ be locally surjective homomorphisms. Then there exists a
  graph $H=(V,E)$ and an RSP-relation $R$ on $E$ with equivalence classes
  $\vp,\ \vpo$ such that
\begin{equation*}
  V=V(G_1)\cup V(G_2) \quad\text{\ and\ }\quad \varphi=E(G_1)\cup
  E(G_2). 
\end{equation*}
Note, it is allowed to have $G_1=G_2$. In this case, $H$ might have loops
and double edges.  
\end{lemma}

\begin{proof}
  For given graphs $G_1$, $G_2$, $G$ and locally surjective homomorphisms
  $f_i:G\rightarrow G_i$, $i=1,2$ construct the graph $H$ as follows: For
  $x\in V(G_1)$ and $y\in V(G_2)$ add an edge $[x,y]$ if and only if there
  exists $g\in V(G)$ such that $f_1(g)=x$ and $f_2(g)=y$.  We set
  $[x,y]\in\vpo$.  It is clear, that $R$ is an equivalence relation.  We
  have to show, that $R$ is an RSP-relation.  Let $[x_1,x_2]\in E(G_1)$ and
  $[x_1,y_1]$ be an added edge. Then there exists $g_1\in V(G)$, such that
  $f_1(g_1)=x_1$ and $f_2(g_1)=y_1$. Since $f_1$ is a locally surjective
  homomorphism, there exists a vertex $g_2$ as a neighbor of $g_1$, such
  that $f_1(g_2)=x_2$. Let $y_2=f_2(g_2)$. Then $y_2$ and $x_2$ are
  connected by an added edge and $y_1,y_2$ are adjacent since $f_2$ is a
  homomorphism.  Thus $[x_1,x_2]$ and $[x_1,y_1]$ lie on a square with
  opposite edges in relation $R$.
 
  If $G_1=G_2$, then just identify vertices of two copies of $G_1$. 
\end{proof}

\begin{lemma}\label{lem:covereq}
Let $G$ and $G'$ be two graphs.
Then there exists a graph $H=(V,E)$ and a well-behaved RSP-relation $R$
with two equivalence classes $\varphi,\ \vpo$ such that
\begin{equation*}
  V=V(G)\cup V(G') \text{\ and\ } \varphi=E(G)\cup E(G'), 
  \text{\ and each vertex of\ } V(G) \text{\ is incident to exactly one\ } 
\vpo\text{-edge} 
\end{equation*}
if and only if $G$ is a cover of $G'$.
\end{lemma}

\begin{proof}
  Let $H=(V,E)$ be a graph with well-behaved RSP-relation $R$ on $E$ as
  claimed.  Then, we can consider $G, G'$ as $\vp$-layer.  By
  Lemma~\ref{lem:cover}, $C_{G',G}$ is a cover of $G'$ and $G$. Since each
  vertex in $V(G)$ is incident with exactly one $\vpo$-edge, we see that
  for covering map $f_1: C_{G',G}\rightarrow G$ holds $|f_1^{-1}(u)|=1$ for
  all $u\in H$ which implies $f_1$ is also injective, thus an isomorphism.

  For the converse, assume $G$ is a cover of $G'$. Then $G$ is a cover of
  $G$ and $G'$ and thus $G$ and $G'$ can be connected as in the prove of
  Lemma \ref{lem:cover3}.  Since clearly $G\cong G$ and thus the covering
  map $p: G\rightarrow G$ is in particular injective, each vertex is, by
  construction, incident to exactly one $\vpo$-edge.  This in turn implies,
  $H$ contains no square $w-x-y-z$ such that $z\in V(G)$ and
  $[w,z],[y,z]\in\vpo$.  On the other hand, there is no square $w-x-y-z$
  contained in $H$ with $[w,x],[x,y]\in E(G)\subseteq \vp$ and
  $[w,z],[y,z]\in\vpo$, i.e., $z\in V(G')$, since otherwise the restriction
  of the covering map $p':G\rightarrow G'$ to $N_G(x)$ (w.l.o.g. we can
  assume $p$ to be the identity mapping) would not be injective, a
  contradiction.  Hence, we can conclude that $R$ is well-behaved.
\end{proof}

Notice that checking if $H$ is a cover graph of $G$ is in general NP-hard
\cite{Abello:91}.  Therefore, also connecting two graphs as described in
Lemma \ref{lem:covereq} is NP-hard.  On the other hand, one can connect two
arbitrary graphs $G_1$, $G_2$ such that all vertices of $G_1$ are linked to
all vertices of $G_2$.  Then, the relation defined by the classes $\vp =
E(G_1)\cup E(G_2)$ and $\vpo$ that consists of all added edges between
$G_1$ and $G_2$ is an RSP-relation. This implies that any two graphs have a
common finite quasi-cover.  However, this is not true for covers, just take
$K_2$ and $K_3$ as an example.

For a given graph $G$ and an RSP-relation $R$, one can consider the
subgraph $G_{\varphi}$, $\vp\eqcl R$ as one layer and all other edges of
$G$ not contained in $G_{\varphi}$ as connecting edges.  Notice,
connectivity is not explicitly needed in Definition~\ref{def:superimposed}
and Lemma~\ref{lem:cover2}, and thus, they can be extended to
$C_{G_\vp,G_\vp}$.  Moreover, any spanning subgraph $H$ of a graph $G$
induces an equivalence relation $R$ with two equivalence classes $E(H)$ and
$E(G)\setminus E(H)$.  Hence, $C_{H,H}$ is well defined and thus,
Lemma~\ref{lem:cover2} and \ref{lem:cover3} imply the following result.

\begin{theorem}
  A graph $G$ has an RSP-relation with two equivalence classes if and only
  if there exists a (possibly disconnected) spanning subgraph $H\subsetneq
  G$ and $C_{H,H}$ is a quasi-cover of $H$.
\end{theorem}

On the set of graphs $\mathfrak{G}$ we consider the relation $G_1\thicksim
G_2$ if $G_1$ and $G_2$ have a common finite cover.
\begin{theorem}
  The relation $\thicksim$ on $\mathfrak{G}$ is an equivalence relation.
 \label{thm:equivRel}
\end{theorem}

\begin{proof}
  Relation $\thicksim$ is clearly reflexive and symmetric. By assumption,
  the graphs $G_1$ and $G_2$ have a common cover $H_{12}$ and $G_2$ and
  $G_3$ have a common cover $H_{23}$. By Lemma \ref{lem:covereq}, $H_{12},
  G_2$ and $H_{23},G_2$ can be connected without forbidden colorings of
  $K_{2,3}$.  Let $E$ be the set of all edges connecting $G_2$ and $H_{12}$
  and $E'$ edges connecting $G_2$ and $H_{23}$. Since every cover of
  $H_{12}$ and $H_{23}$ is a cover of $G_1$, $G_2$ and $G_3$, it is
  sufficient to find a cover of $H_{12}$ and $H_{23}$. Therefore, it
  suffices to connect $H_{12}$ and $H_{23}$ without forbidden colorings of
  $K_{2,3}$. Define edges connecting $H_{12}$ and $H_{23}$ by connecting
  $h\in V(H_{12})$ and $h'\in V(H_{23})$ if there exists a vertex $v\in
  V(G_2)$ such that $[h,v]\in E$ and $[v,h'] \in E'$.
  
  First we check that $E(H_{12})\cup E(H_{23})$ and connecting edges form
  two equivalence classes of an RSP relation. Without loss of generality
  assume $[h_1,h_2]\in E(H_{12})$ and $[h_1,h_1']$, $h_1'\in V(H_{23})$ is
  a connecting edge. Then there exists $v_1\in V(G_2)$ such that $[h_1,v_1]
  \in E$ and $[v_1,h_1'] \in E'$. Since edges $E$ are defined by a local
  bijection between $H_{12}$ and $G_2$, there exist $v_2\in V(G_2)$, a
  neighbor of $v_1$, such that $[h_2,v_2]\in E$. Similarly, since $E'$ is
  defined by a local bijection between $H_{23}$ and $G_2$, there exists
  $h_2'\in V(H_{23})$, a neighbor of $h_1'$, such that $[v_2,h_2']\in
  E'$. Therefore there exists a square $h_1-h_1'-h_2'-h_2$ with
  $[h_1,h_2],[h_1',h_2']\in E(H_{12}) \cup E(H_{23})$ and
  $[h_1,h_2],[h_1',h_2']$ being connecting edges. This proves that relation
  $R$, with equivalence classes $E(H_{12})\cup E(H_{23})$ and the set of
  connecting edges is an RSP relation.
  
  It remains to prove that $R$ is well-behaved. By symmetry, it is enough
  to prove that there exists no vertices $h_1,h_2,h_3 \in V(H_{12})$ and
  $h_1',h_2'\in V(H_{2,3})$ with $[h_1,h_2],[h_1,h_3]\in E(H_{12})$,
  $[h_1',h_2'] \in E(H_{23})$ and added edges $[h_1,h_1'], [h_2,h_2']$ and
  $[h_3,h_2']$.  For the sake of contradiction, assume such vertices
  exist. By the construction of the added edges, there exist vertices $v_1,
  v_2,v_3 \in V(G_2)$ such that $[h_1,v_1],[h_2,v_2], [h_3,v_3] \in E$ and
  $[v_1, h_1'], [v_2,h_2'], [v_3,h_2'] \in E'$.  Since edges in $E$ are
  obtained from a covering map of $H_{12}$ to $G_2$ we see that $v_1,v_2$
  and $v_3$ are distinct vertices.  But also the edges in $E'$ are obtained
  from a covering map of $H_{23}$ to $G_2$ therefore
  $[v_2,h_2']=[v_3,h_2']$ and thus $v_2=v_3$, a contradiction.
\end{proof}
We have proven Theorem~\ref{thm:equivRel} here by elementary means to keep
this presentation self-contained. It also follows from a deep result of
Leighton \cite{leighton1982finite}, who proved the following: A pair of
finite connected graphs $G_1$ and $G_2$ has a common finite cover if and
only if they have the same (possibly infinite) cover graph isomorphic to a
tree. Such a cover is unique for every graph $G$ and covers any other
covering graph of $G$; It is therefore called the \emph{universal cover} of
$G$. On the other hand, a minimal common cover of two graphs needs not to
be unique, as Imrich and Pisanski have shown \cite{imrich2008multiple}.

\begin{coro}
  Let $G$ be a connected graph and let $R$ be a well-behaved RSP-relation
  on $E(G)$. Then there exists a common covering graph for all
  $\varphi$-layer $G_{\varphi}^{x_i}$.
\end{coro}
\begin{proof}
  This result is an immediate consequence of the connectedness of $G$,
  Lemma \ref{lem:cover} and Theorem \ref{thm:equivRel}.
\end{proof}
In terms of Leighton's theorem, the corollary could be read in the
following way: For a graph $G$ with a well-behaved RSP-relation on $E(G)$
and some fixed equivalence class $\varphi$ all the graphs
$\{G_{\varphi}^{x_i}\}$ have the same universal cover.

Under certain conditions it is possible to refine a given RSP-relation.
\begin{lemma}
\label{lem:split1}
Let $G=(V,E)$ be a connected graph and $R$ a well-behaved RSP-relation on $E$. 
Assume that for one equivalence class $\vp\eqcl R$ the graph $G_{\vp}$ has
two connected components $G_{\vp}^x$ and $G_{\vp}^y$.
The next two statements are equivalent:
\begin{enumerate}
\item There is a well-behaved refined RSP-relation $R' \subsetneq R$ such that 
$\vp=\chi_1\cup \chi_2$ with $\chi_1,\chi_2\eqcl R'$
\item $C_{G_{\vp}^x,G_{\vp}^y}$ has a non-trivial RSP-relation $Q$
such that $(e,f)\in Q$ iff $(e',f')\in R'$ for all 
$e,f \in p_1^{-1}(e') \cup p_1^{-1}(f') \cup  p_2^{-1}(e') \cup p_2^{-1}(f')$ 
and for all $e,f\in E(G_{\vp}^x) \cup E(G_{\vp}^y)$, where 
$p_1:C_{G_{\vp}^x,G_{\vp}^y} \rightarrow G_{\vp}^x$, resp., 
$p_2:C_{G_{\vp}^x,G_{\vp}^y} \rightarrow G_{\vp}^y$. 
\end{enumerate}
In other words, $R$ can be refined to $R'$ if and only if edges of
$G_{\vp}^x$, resp., $G_{\vp}^y$ that map on the same edges via the covering
projection are in the same class w.r.t.\ $Q$.
\end{lemma}
\begin{proof}
  If there is a finer RSP-relation $R'$, every square $a_1-b_1-b_2-a_2$ with
  $a_1,a_2\in V(G_{\vp}^x)$ and $b_1,b_2\in V(G_{\vp}^y)$ has edges
  $[a_1,a_2]$ and $[b_1,b_2]$ in the same class by the relaxed square property
  and since $R$ is well-behaved. Thus, an equivalence relation on $E(G_{\vp}^x)$ and
  $E(G_{\vp}^y)$ can be lifted to an equivalence relation on $E(C_{G_{\vp}^x,G_{\vp}^y})$ in a
  natural way.  One can check that it has the relaxed square property by using that the respective relations on
  $E(G_{\vp}^x)$ and
  $E(G_{\vp}^y)$ have the relaxed square property.

  Conversely, we define a finer RSP-relation on $E(G_{\vp}^x)$ and $E(G_{\vp}^x)$
  from the RSP-relation on $E(C_{G_{\vp}^x,G_{\vp}^y})$ by setting $(e',f')\in R'$
  iff $(e,f)\in Q $ for some $e\in p_1^{-1}(e'), f\in p_1^{-1}(e')$.
\end{proof}

Let $R$ be a well-behaved RSP-relation on $G$, e.g., $R=\delta_0$, and
suppose there is a finer RSP-relation $R'$ in which an equivalence class
$\varphi$ is split into two equivalence classes $\varphi_1$ and
$\varphi_2$. Let $\{G_{\varphi}^{x_i}\}$ be the connected components of
$G_{\varphi}$. Then $\varphi_1$ and $\varphi_2$ induce an RSP-relation on
each $G_{\varphi}^{x_i}$. Consider two components $G_{\varphi}^{x_1}$ and
$G_{\varphi}^{x_2}$ that are connected by some edges (in other
classes). From the proof of Lemma \ref{lem:split1} we observe that an
RSP-relation on $E(G_{\varphi}^{x_1})$ already defines an RSP-relation on
$C_{G_{\varphi}^{x_1},G_{\varphi}^{x_2}}$, which in turn defines an
RSP-relation on $G_{\varphi}^{x_2}$ and thus on all $\vp$-layer
$G_{\varphi}^{x_i}$. If multiple splits of $\varphi$ exist, they are fixed
by choosing one on any $G_{\varphi}^{x_i}$.
 
Now consider the graph $G$ consisting of two copies of $K_{2,3}$ and all
edges connecting them and the equivalence relation whose two classes are
the edges of the two copies of $K_{2,3}$ and the connected edges,
respectively. The discussion above implies that we can split the first
class independently on the two copies of $K_{2,3}$. Thus, we cannot
generalize the result above to RSP-relations with forbidden colorings. 

\section{Outlook and Open Questions}

We discussed in this contribution in detail RSP-relations, the most relaxed
type of relations fulfilling the square property. As it turned out, such
relations are hard to handle in graphs that contain $K_{2,3}$-subgraphs. On
the other hand, it is possible to determine finest RSP-relation in
polynomial time in $K_{2,3}$-free graphs.  Moreover, we showed how to
determine (finest) RSP-relations in certain graph products, as well as in
complete and complete-bipartite graphs. We finally established the close
connection of (well-behaved) RSP-relations to graph covers and equitable
partitions. Intriguingly, non-trivial RSP-relations can be characterized by
means of the existence of spanning subgraphs that yield quasi-covers of the
graph under investigation.

Still, many interesting problems remain open topics for further research.
From the computational point of view, it would be worth to determine the
complexity of the problem of determining finest (well-behaved)
RSP-relation.  Since there is a close connection to graph covers, we
suppose that the latter problem is NP-hard. If so, then fast heuristics
need to be designed. It is also of interest to investigate, for which graph
classes (that are more general than $K_{2,3}$-free graphs) the proposed
algorithm determines well-behaved or finest RSP-relations.

From the mathematical point of view, one might ask, under which
circumstances is it possible to guarantee that there is a non-trivial
finest RSP-relation that is in addition well-behaved.  Note, the graph
$G=K_{2,3}$ has no such relation. However, there might be interesting graph
classes that have one.  In addition, it might be of particular importance
(also for computational aspects) to distinguish RSP-relation.  Let us say
that two RSP-relations $R$ and $S$ on $E$ are \emph{equivalent}, $R\simeq
S$, if there is an automorphism $f:V\to V$ such that $((x,y),(a,b))\in R$
if and only if $((f(x),f(y)),(f(a),f(b)))\in S$.  Note, if $G=K_{2,3}$ then
all finest RSP-relation consist of two equivalence classes and all such
relations are equivalent.

Clearly, if $R\simeq S$, then $G/\mc{P}^R \simeq
G/\mc{P}^S$. However, the converse is not true, i.e., $G/\mc{P}^R \simeq
G/\mc{P}^S$ does not imply $R\simeq S$, see Example \ref{ex:k9}. This
suggests to consider under which conditions finest RSP-relations are
unique or for which graphs the equivalence of RSP-relations can be expressed in
terms of isomorphism of quotient graphs. 

\bibliographystyle{plain}  
\bibliography{rsp}
 
\end{document}